\documentclass[a4paper, 11pt]{article}                      
\usepackage{pdfsync}
\usepackage{amsthm,hyperref}
\usepackage{mathrsfs}
\usepackage{supertabular}
\usepackage{cite}
\usepackage{color}
\usepackage{paralist}
\usepackage{graphics} 
\usepackage{latexsym,  amsfonts, amssymb,graphicx,array,tabularx}
\usepackage{subfigure}
\usepackage{caption}
\usepackage{dsfont}
\usepackage{bbm}
\usepackage{epsfig}
\usepackage{amsmath}
\usepackage[english]{babel}
\usepackage{multicol}
\usepackage{paralist}
\usepackage{mathbbold, bbm}
\usepackage{mathtools}
\DeclareFontFamily{OT1}{pzc}{}
\DeclareFontShape{OT1}{pzc}{m}{it}{<-> s * [1.000] pzcmi7t}{}
\DeclareMathAlphabet{\mathpzc}{OT1}{pzc}{m}{it}

\theoremstyle{theorem}

\newtheorem{theorem}{Theorem}
\newtheorem{lemma}{Lemma}

\newtheorem{remark}{Remark}

\newtheorem{definition}{Definition}
\newtheorem{assumption}{Assumption}
\newtheorem{problem}{Problem}

\DeclareMathAlphabet{\mathpzc}{OT1}{pzc}{m}{it}

\newcommand{\mcal}{\mathcal}
\newcommand{\R}{{\rm  I\kern-2pt R}}
\renewcommand{\Re}{{\rm  I\kern-2pt R}}

\newcommand{\argmin}{\textrm{arg}\min}

\definecolor{awesome}{rgb}{1.0, 0.13, 0.32}

\newcommand\addtag{\refstepcounter{equation}\tag{\theequation}}

\makeatletter
\newcommand{\rmnum}[1]{\romannumeral #1}
\newcommand{\Rmnum}[1]{\expandafter\@slowromancap\romannumeral #1@}
\makeatother

\begin{document}

\title{\bf Infinite Horizon Optimal Transmission Power Control \\ for Remote State Estimation over Fading Channels}
\author{Xiaoqiang Ren,
Junfeng Wu, Karl Henrik Johansson, Guodong Shi, and Ling Shi\thanks{X. Ren and  L. Shi
 are with the Department of Electronic and Computer Engineering, The Hong Kong University of Science and Technology,
 Hong Kong. Email: xren@connect.ust.hk, eesling@ust.hk}
\thanks{J. Wu and K. H. Johansson are with the ACCESS Linnaeus Center, School of Electrical Engineering,
Royal Institute of Technology, Stockholm, Sweden. Email: junfengw@kth.se, kallej@kth.se}
\thanks{G. Shi is with College of Engineering and Computer Science, The Australian National University, Canberra, Australia. Email: guodong.shi@anu.edu.au}
}
\date{}
\maketitle

\begin{abstract}
Jointly optimal transmission power control and remote estimation over an infinite horizon is studied. A sensor observes a dynamic process and sends its observations to a remote estimator over a wireless fading channel characterized by a time-homogeneous Markov chain. The successful transmission probability depends on both the channel gains and the transmission power used by the sensor. The transmission power control rule and the remote estimator should be jointly designed, aiming to minimize an infinite-horizon cost consisting of the power usage and the remote estimation error. A first question one may ask is: {Does this joint optimization problem
have a solution?}
We formulate the joint optimization problem as an average cost
belief-state Markov decision process and answer the question by proving that there exists an optimal deterministic and stationary policy. We then show that when the monitored dynamic process is scalar, the optimal remote estimates depend only on the most recently received sensor observation, and the optimal transmission power is symmetric and monotonically increasing with respect to the innovation error.
\end{abstract}



\section{Introduction}

In networked control systems, control loops are
often closed
over a shared wireless communication network.
This motivates research on remote state estimation problems, where
a sensor measures the state of a linear system and transmits its observations to a remote estimator over a wireless fading channel. Such monitoring problems appear in a wide range of
applications in environmental monitoring, space exploration, smart grids,
intelligent buildings, among others.
The challenges introduced by the networked setting lie
in the fact that
nonideal communication environment and
constrained power supplies at sensing nodes
may result in overall system performance degradation.
The past decade has witnessed tremendous research efforts devoted to  communication-constrained estimation problems,
with the purpose of establishing a balance between
estimation performance and communication cost.

%

\subsection{Related Work}
Wireless communications are being widely used nowadays
in sensor networks and networked control systems.
The interface of control and wireless communication has been
a central theme in the study of networked sensing and control systems in the past decade. Early works assumed finite-capacity digital channels and focused on the minimum channel capacity or data rate
needed for feedback stabilization, and on constructing
 encoder-decoder pairs to improve performance,~e.g., \cite{wong2,nair2004stabilizability,
tatikonda2,ishii, fu-xie-tac05}.
Motivated by the fact that packets are
the fundamental information carrier in most modern data networks~\cite{joao07}, networked control and estimation subject to
packet delays~\cite{schenato2008optimal,shi2009kalman,you2011mean} and packet losses~\cite{sinopoli2004kalman,huang-dey-stability-kf,
schenato2007foundations,gupta2009optimal} has been
extensively studied.

State estimation is embedded
in many networked control applications, playing a fundamental role therein.
For networked state estimation subject to limited communication resource, the research on controlled communication has been extensive, see the survey~\cite{joao07}.
Controlled communication, in general referring to reducing the communication rate intentionally to obtain a desirable tradeoff between the estimation performance and the communication rate, is motivated from at least two facts:
 \begin{inparaenum}
 \item[$(i).$]Wireless sensors are usually battery-powered and sparsely deployed, and replacement of battery is difficult or even impossible, so the amount of communication needs to e kept at a minimum as communication is often the dominating on-board energy consumer~\cite{mainwaring2002wireless}.
 \item[$(ii).$]Traffic congestion in a sensor network many lead to packet losses and other network performance degradation.
  \end{inparaenum}
 To minimize the inevitable enlarged estimation
error due to reduced communication rate,
a communication scheduling strategy for the sensor is needed.
Two lines of research directions are present in the literature. The first line is known as time-based (offline) scheduling, whereby the communication decisions are simply specified only according to the time. Informally, a
purely time-based strategy is likely to lead to a periodic communication
schedule~\cite{yang2011deterministic,zhao2014optimal}.
Optimal periodic scheduling has been extensively studied, e.g,~\cite{shi2011optimal,huber2012optimal,shi2013optimal,liu2014optimal}.
The second line is known as event-based scheduling,
whereby the communication decisions are specified according to the system state.
The idea of event-based scheduling was popularized by
Lebesgue sampling~\cite{aastrom2002comparison}.
Deterministic event-based transmission schedules
have been proposed in~\cite{xu05cdc,Imer-CDC-05,randy-cogill-acc07,sijs2009event,
Lispa11TAC,wu2013event,nayyar2013optimal,ramesh2013design,
junfeng14tac,molin2014optimal} for different application scenarios, and randomized event-based transmission schedules can be found in~\cite{vijay_sensor_schedule,han2014stochastic,
journals/corr/WeerakkodyMSHS15}. Essentially, event-based scheduling is a sequential decision problem with
a team of two agents (a sensor and an estimator).
Due to the nonclassical information structure of the two agents,  joint optimization of the communication controller and
the estimator is hard~\cite{yuksel2013stochastic}.
Most works~\cite{xu05cdc,Imer-CDC-05,randy-cogill-acc07,sijs2009event,wu2013event,nayyar2013optimal,ramesh2013design,
junfeng14tac,molin2014optimal}
bypassed the challenge by imposing
restricted information structures or by approximations,
while some authors
have obtained structural descriptions of the agents under the joint optimization framework, using a majorization argument~\cite{Lispa11TAC,nayyar2013optimal} or
an iterative procedure~\cite{molin2014optimal}.
In all these works communication models were highly simplified, restricted to a binary switching model.

Fading is non-ignorable impairment to wireless communication~\cite{goldsmith2005wireless}. The effects of fading has been taken into account in networked control systems~\cite{elia2005remote,dey2009kalman,xiao2012feedback,Quevedo13TAC}.
There are works that are concerned with transmission power
management for state estimation~\cite{wang2009distortion,queahl10,leong2011power,
WuAutomatica13,gatsis2014optimal,nourian2014optimal,Nourian15JSAC}.
The power allocated
to transmission affects the probability of successful reception of the measurement, thus affecting the estimation performance.
In~\cite{queahl10}, transmission power is allocated
via a predictive control algorithm
based on the channel gain prediction.
In \cite{nourian2014optimal}, imperfect acknowledgments of communication links and energy harvesting were taken into account.
In~\cite{leong2011power},
power allocation for the estimation
outage minimization problem was investigated in estimation of a scalar Gauss-Markov source.
In all of the aforementioned works,
the estimation error covariances are a Markov chain controlled by the transmission power, so
Markov decision process (MDP) theory is ready for solving this kind of problems.
The reference~\cite{gatsis2014optimal} considered
the case when
plant state is transmitted from a sensor to
the controller over a wireless fading channel. The transmission power is adapted to the
channel gain and the plant states.
Due to nonclassical information structure, joint optimization of plant input and
transmit power policies, although desired, is difficult.
A restricted information structure was therefore imposed,
i.e., only a subset of the full information history available at the sensor is utilized when determining the transmission power,
to allow separate design at expense of loss of optimality.
It seems that such a challenge involved in these joint optimization
problems always exists.

\subsection{Contributions}

In this paper, we consider a remote state estimation scheme, where
a sensor measures the state of a linear time-invariant discrete-time process and transmits its observations to a remote estimator over a  wireless fading channel characterized by a time-homogeneous Markov chain. The successful
transmission probability depends on both the channel gain and
the transmission power used by the sensor.
 The objective is
to minimize an infinite horizon cost consisting of the power
consumption and the remote estimation error.
In contrast to~\cite{gatsis2014optimal},
no approximations are made to prevent loss of optimality, which however renders the analysis challenging.
We
formulate our problem as an infinite horizon belief-state MDP with an average cost criterion.
Contrary to the finite horizon belief-state MDP considered
in~\cite{nayyar2013optimal}, for which an optimal solution
exists, a first question that one may ask about our infinite horizon MDP is: \emph{Does this optimization problem have a solution?}
The answer is yes provided certain conditions given in this paper.
On top of this, we present structural results on the optimal transmission power controller and the remote estimator for some special systems, which can be seen
as the extension of the results in~\cite{Lispa11TAC,molin2014optimal} for the power management scenario.
The analysis tools used in the work (i.e., the partially observable Markov decision process (POMDP) formulation  and the majorization interpretation) is inspired
by~\cite{nayyar2013optimal}. Nevertheless, the contributions of the two works are distinct. In~\cite{nayyar2013optimal} the authors mainly studied the
threshold structure of the optimal communication strategy
within a finite horizon, while the present work focuses on the
asymptotic analysis of the  joint optimization problem over an
infinite horizon.

In summary, the main contributions of this paper are listed as follows. We prove that a deterministic and stationary policy is an optimal solution to the formulated average cost belief-state MDP.
      We should remark that the abstractness of the considered state and action spaces (the state space is  a probability measure space and the action space a function space) renders the analysis rather challenging. Then we prove that both the optimal estimator and the optimal power control have simple structures when the dynamic process monitored is scalar. To be precise, the remote estimator synchronizes its estimates with the data received in the presence of successful transmissions, and linearly projects its estimates a step forward otherwise. For a certain belief, the optimal transmission power is a symmetric and monotonically increasing function of the innovation error. Thanks to these properties, both the offline computation and the online implementation of the optimal transmission power rule are greatly simplified, especially when the available power levels are discrete, for which only thresholds of switchings between power levels are to be determined.

This paper provides a theory in support of the study of infinite horizon communication-constrained estimation problems.
Deterministic and stationary policies are relatively easy to compute and implement, thus it is important to know that
an optimal solution  that such a policy exists.
The structural characteristic of the jointly optimal transmission power and estimation policies provides insights into  the design of energy-efficient state estimation algorithms.


\subsection{Paper Organization}
In Section~\ref{sec:problem-setup}, we provide the mathematical formulation of the system model adopted, including the monitored dynamic process, the wireless fading channel, the transmission power controller and the remote estimator.
We then present the considered problem and formulate it as an average cost MDP in Section~\ref{sec:remote-state-estim}. In Section~\ref{sec:ExistOptPolicy}, we prove that there exists a deterministic and stationary policy that is optimal to the formulated MDP. Some structural results about the optimal remote estimator and the optimal transmission power control strategy are presented in Section~\ref{sec:structureResults}. Concluding remarks are given in Section~\ref{sec:Conclusion}. Some auxiliary background results and a supporting lemma are provided in the appendices.

\subsection*{Notation} 
$\mathbb{N}$ and $\mathbb{R}_+$ are the sets of
 nonnegative integers and
nonnegative real numbers, respectively.
$\mathbb{S}_{+}^{n}$ (and $\mathbb{S}_{++}^{n}$)
is the set of $n$ by $n$ positive semi-definite
matrices (and positive definite matrices).
When $X \in\mathbb{S}_{+}^{n}$ (and $\mathbb{S}_{++}^{n}$),
we write $X \succeq 0$ (and $X \succ 0$). $X\succeq Y$
if $X - Y \in \mathbb{S}_{+}^{n}$. $\mathrm{Tr}(\cdot)$
and $\mathrm{det}(\cdot)$ are the trace and the determinant
of a matrix, respectively.
$\lambda_{\rm max}(\cdot)$ represents the eigenvalue, having the largest magnitude, of
a matrix.
The superscripts $^\top$ and $^{-1}$
stand for matrix transposition and matrix inversion, respectively.
The indictor function of a set $\mathcal{A}$ is defined as
$$
\mathds{1}_{\mathcal{A}}(\omega)=\left\{\begin{array}{lll}
1, & \omega\in\mathcal{A}\\
0, & \omega\not \in \mathcal{A}.
\end{array}\right.
$$
The notation $p(\mathbf{x};x)$
represents the probability density function (pdf) of a random variable $\mathbf{x}$ taking value at $x$.
If being clear
in the context, $\mathbf{x}$ is omitted.
For a random variable $\mathbf{x}$ and a pdf $\theta$,
the notation $\mathbf{x}\sim \theta$ means
that $\mathbf{x}$ follows the distribution defined by
$\theta$.
The symbol $\mathscr{N}_{x_0,\Sigma}(\mathbf x)$ denotes a Gaussian distribution of $\mathbf x$
with mean $x_0$ and covariance $\Sigma$.
For measurable functions $f, g:\mathbb{R}^n\mapsto\mathbb{R}$, we use $f*g$ to denote the convolution of $f$ and $g$.
For a Lebesgue measurable set $A \subset \mathbb{R}^n$, $\mathfrak{L}(A)$ denotes the Lebesgue measure of $A$. Let $\|x\|$ denote the $L^2$ norm of a vector $x\in\mathbb{R}^n$.
$\delta_{ij}$ is the Dirac delta function, i.e., $\delta_{ij}$ equals to $1$ when $i=j$ and $0$ otherwise.


\section{System Model} \label{sec:problem-setup}
In this paper, we focus on  dynamic power control for remote state estimation. We consider a remote state estimation scheme as depicted in Figure~\ref{fig:system}. In this scheme, a sensor
measures a linear time-invariant discrete-time process and sends its measurement in the form of data packets, to
 a remote estimator over a wireless link.
 The remote estimator produces an estimate
of the process state based on the received data. When sending packets through
the wireless channel,
transmissions may fail due to interference and weak channel gains.
Packet losses lead to distortion of the remote estimation and packet loss probabilities depend on transmission power levels used
by the transmitter and on the channel gains.
Lower loss probabilities require
higher transmission power usage; on the other hand,
energy saving is critical to expand
the lifetime of the sensor. The wireless communication
overhead dominates the total power consumption, therefore we
introduce
a transmission power controller, which aims to
balance the transmission energy cost
and distortion penalty as the
channel gain varies over time.

In what follows, the attention is devoted to laying out the main
components in Figure~\ref{fig:system}.

\begin{center}
\setlength{\unitlength}{1.5mm}
\begin{figure}[t!]
\thicklines
\centering
\begin{picture}
(56,28)(-3,-9)
\thicklines
\put(0,14){\framebox(10,4){${\rm Process}$}}
\put(3,7.8){\vector(1,0){5}}
\put(3,14){\vector(0,-1){12}}
\put(0,-4){\framebox(13,6){${\rm Tx~Power\atop Controller}$}}
\put(10.5,2){\vector(0,1){4}}
\put(11.5,3.5){$u_k$}
\put(8,6){\framebox(5,3.5){${\rm Tx}$}}
\put(13,7.75){\vector(1,0){8.5}}
\put (17.5,8.75){$x_k$}
\put(21.5, 4.75){\framebox(8,6)}
\put(38,5.25){\framebox(9,5){${\rm Remote \atop Estimator}$}}
\put(47,7.75){\vector(1,0){7}}
\put (48,9){$\tilde x_k$}
\put (21.75,7.1){$\rm Wireless\atop Channel$}
\put(29.5,7.75){\vector(1,0){8.5}}
\put (30.5,9){$\gamma_k x_k$}
\put(42.5,5.25){\line(0,-1){6.25}}
\put(42.5,-1){\vector(-1,0){29.5}}
\put(-1, -7.5){\framebox(16,19)}
\put (-0.8,-7){$\rm Sensor$}
{\color{blue}
\put(-2.5, -8.5){\dashbox{1}(38.5,27.5)}}
{\color{red}
\put(16, -2.5){\dashbox{2}(34.5,14.5)}}
\end{picture}
\caption{The remote state estimation scheme.}
\label{fig:system}
\end{figure}
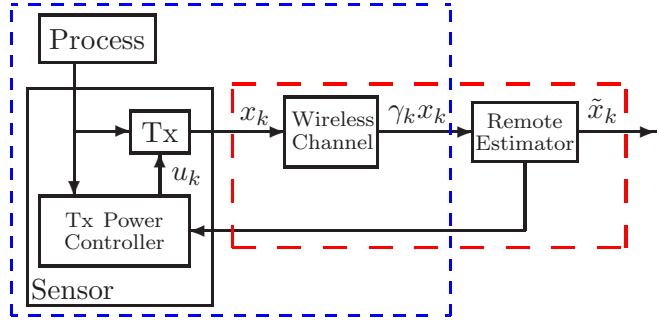
\end{center}
\vspace{-1cm}

\subsection{State Process} \label{sec:local-state-estimate}


We consider the following linear time-invariant discrete-time process:
\begin{equation}
  x_{k+1}  = Ax_k + w_k, \label{eqn:process-dynamics}
\end{equation}
where  $k\in \mathbb{N}$, $x_k \in \mathbb{R}^{n}$ is the process state vector at time $k$, $w_{k} \in\mathbb{R}^{n} $ is zero-mean independent and identically distributed (i.i.d.) noises, described by the probability density function (pdf)
$\mu_w$,  with $\mathbb{E}[w_{k}w_{k}^\prime] =W$ ($W\succ 0$).
The initial state $x_0$, independent of
${w_k}, k\in\mathbb{N}$, is described by pdf $\mu_{x_0}$, with mean $\mathbb{E}[x_0]$ and covariance $\Sigma_0$.
Without loss of generality, we assume $\mathbb{E}[x_0]=0$, as
nonzero-mean cases can be translated into zero-mean one
by coordinate change
$ x'_k=x_k-\mathbb{E}[x_0]$.
We let
$\mathcal{X}=\mathbb{R}^n$ denote the domain of $x_k$.
The system parameters are all known to the sensor as well as the remote estimator. We assume the plant is unstable, i.e., $|\lambda_{\mathrm{max}}(A)| > 1$.

\subsection{Wireless Communication Model }\label{subsec:wireless-communcation-model}

The sensor measures and sends the process state $x_k$ to the remote estimator over an additive white Gaussian noise (AWGN) channel
which suffers from channel fading (see Figure~\ref{wireless-channel-model}):
$$
\mathbf{y}=g_k \mathbf{x}+v_k,
$$
where $g_k$ is a random complex number, and
$v_k$ is additive white Gaussian noise; $\bf x$ represents the signal (e.g., $x_k$) sent by the transmitter and $\bf y$ the signal
received by the receiver. Let the channel gain $h_k = |g_k|^2$ take values in a finite set $\mathbbm{h} \subseteq \mathbb{R}_+$
and
$\{h_k\}_{k\in\mathbb{N}}$ possess temporal correlation modeled by a time-homogenous Markov chain~\cite{Quevedo13TAC,Nourian15JSAC}. The one-step transition probability for this chain is denoted by
\[
\pi(\cdot|\cdot): \mathbbm{h}\times \mathbbm{h} \longmapsto
[0,1].
\]
The function $\pi(\cdot|\cdot)$ is known \textit{a priori}.
We assume the remote estimator or the sensor can access
the channel state information (CSI), so
the channel gain $h_k$ is available
at each time before transmission.
This can for instance done using
pilot-aided channel estimation techniques are adopted, by which the transmitter sends a pilot signal at each fading block and the channel coefficients, including the channel gain,
are obtained at the receiver~\cite{goldsmith2005wireless}.
The estimation errors of the channel gains are not taken into account in this paper.
\begin{center}
\setlength{\unitlength}{1.5mm}
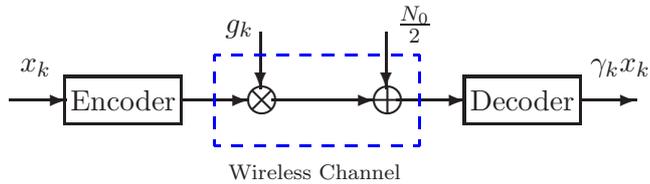
\begin{figure}[t!]
\thicklines
\centering
\begin{picture}
(56,16)(0,-8)
\thicklines
\put(5,-2){\framebox(10,4){${\rm Encoder}$}}
\put(40,-2){\framebox(10,4){${\rm Decoder}$}}
\put(1,2.5){$x_k$}
\put(51,2.5){$\gamma_k x_k$}
\put(0,0){\vector(1,0){5}}
\put(50,0){\vector(1,0){5}}
\put(23.05,0){\vector(1,0){9}}
\put(34,0){\vector(1,0){6}}
\put(15,0){\vector(1,0){6}}
\put(20.75,-0.625){$\bigotimes$}
\put(31.75,-0.625){$\bigoplus$}
\put(22.05,6){\vector(0,-1){5}}
\put(33.05,6){\vector(0,-1){5}}
\put(19,6){$g_k$}
\put(34,6){$\frac{N_0}{2}$}
{\color{blue}
\put(18, -4){\dashbox{1}(18,8)}
}
\put (19.25,-7){\scriptsize $\rm Wireless~Channel$}
\end{picture}
\caption{Wireless communication model.
${N}_0$ is the power spectral density of the channel noise $v_k$.}
\label{wireless-channel-model}
\end{figure}
\end{center}
\vspace{-0.5cm}

To facilitate our analysis, the following assumption is made.
\begin{assumption}[Communication model] \label{asmpt:CommunicationModel} \hfill

\begin{enumerate}
\item[(\rmnum{1}).]
The channel gain $h_k$ is independent of the system paremeters.

\item[(\rmnum{2}).]
The channel is block fading, i.e., the channel gain
remains constant during each packet
transmission and varies from block to block.

\item[(\rmnum{3}).]
The quantization effect is negligible and does not effect
the remote estimator.

\item[(\rmnum{4}).]
The receiver can
detect symbol errors\footnote{In practice, symbol errors can be detected via a cyclic redundancy check (CRC) code.}. Only the data reconstructed error-free are regarded as successfully reception.
The receiver perfectly realizes whether the instantaneous communication succeeds or not.

\item[(\rmnum{5}).]
The Markov chain governing the channel gains, $\pi(\cdot|\cdot)$, is aperiodic and irreducible.

\end{enumerate}
\end{assumption}
Assumption~\ref{asmpt:CommunicationModel}-(\rmnum{1})(\rmnum{2})(\rmnum{3})(\rmnum{4}) are standard for fading channel model. Note that Assumption \ref{asmpt:CommunicationModel}-(\rmnum{1})(\rmnum{3})(\rmnum{4}) were used in~\cite{sinopoli2004kalman,fu2009Automatica,leong2012power,
WuAutomatica13,gatsis2014optimal}, and that Assumption \ref{asmpt:CommunicationModel}-(\rmnum{2}) was used in~\cite{queahl10}.
From Assumption \ref{asmpt:CommunicationModel}-(\rmnum{4}), 
whether or not the data sent by the sensor is successfully received by the remote estimator is indicated by a sequence $\{\gamma_k\}_{k\in\mathbb{N}}$ of random
variables, where
\begin{equation} \label{eq:43}
  \gamma_k =
  \begin{cases}
  1, &\text{if ${x}_k$ is received error-free at time $k$,}\\
  0, & \text{otherwise (regarded as dropout),}
  \end{cases}\vspace{-1mm}
\end{equation}
initialized with $\gamma_0=1$.
When $\gamma_k=0$, we regard the remote estimator as having received
a virtual symbol $\sharp$. Assumption~\ref{asmpt:CommunicationModel}-(\rmnum{5}) is a technical requirement for Theorem~\ref{theorem:ExistenceStationary}. One notes that both the i.i.d. channel gains model and the Gilbert--Elliott model with the good/bad state transition probability not equal to 1 satisfies Assumption~\ref{asmpt:CommunicationModel}-(\rmnum{5}).

\subsection{Transmission Power Controller}
Let $u_k\in \mathbb{R}_{+}$ be the transmission power at time $k$, the  power supplied to the radio transmitter.
Due to constraints with respect to radio power amplifiers, the admissible transmission power is restricted. Let $u_k$ take values in $\mathcal{U}\subset[0,u_{\mathrm{max}}]$,
in which $u_{\mathrm{max}}$ stands for the maximum power.
Depending on the radio implementation, $\mathcal{U}$ may be
a continuum or a finite set. It is further assumed that $\mathcal{U}$ is compact and contains zero.
Under Assumption~\ref{asmpt:CommunicationModel}-(\rmnum{3}), the successful packet reception
is statistically determined by the signal-to-noise ratio (SNR) ${h_kp_k}/{N_0}$ at the receiver, where
${N}_0$ is the power spectral density of $v_k$.
A very general model~\cite{Quevedo13TAC, gatsis2014optimal} of the conditional packet reception
probabilities for a variety of modulations is as follows:
\begin{equation} \label{definition:lambda_k}
  q(u_k,h_k) \triangleq {\mathbb{P}} \left(\gamma_k=1|
  u_k,h_k\right),
\end{equation}
where $q$ is a nondecreasing function in both $u_k$ and $h_k$.

\begin{assumption}
The function $q(u,h): \mathcal{U}\times \mathbbm{h}\mapsto [0,1]$ is continuous almost everywhere with respect to $u$ for any fixed $h$. Especially, $q(0,h) = 0$ for all $h$.
\end{assumption}
\begin{remark}
If letting $q(u_k,h_k)=q(u_k)$ with $\mathcal{U}=\{0,1\}$ and
 $$q(u_k)=\left\{\begin{array}{ccc} 1, \hbox{~if~}u_k=1;\\
 0, \hbox{~if~}u_k=0,\end{array}\right.$$
we conclude that  the ``on-off'' controlled communication problem considered in~\cite{shi2011optimal,huber2012optimal,shi2013optimal,liu2014optimal,
 xu05cdc,Imer-CDC-05,randy-cogill-acc07,sijs2009event,
Lispa11TAC,wu2013event,nayyar2013optimal,ramesh2013design,
junfeng14tac,molin2014optimal,vijay_sensor_schedule,han2014stochastic,journals/corr/WeerakkodyMSHS15}
is a special case of the transmission power control problem considered here.

\end{remark}
We assume that packet reception probabilities are conditionally
independent for given channel gains and transmission power levels,
which is stated in the following assumption.
\begin{assumption}
The following equality holds for any $k\in\mathbb{N}$,
\begin{equation*}
{\mathbb{P}} \left(\gamma_k=r_k,\ldots,\gamma_1=r_1|
  u_{1:k},h_{1:k}\right)=\prod_{j=1}^{k}
  {\mathbb{P}} \left(\gamma_j=r_j|
  u_j,h_j\right).
\end{equation*}
\end{assumption}

{\begin{remark}
Assumption 2 is standard for digital communication over fading channels. Assumption 3 is in accordance with the common sense that
the symbol error rate statistically depends on the
instantaneous SNR at the receiver.
Many digital communication modulation methods are embraced by these assumptions~\cite{mostofi2009drop, queahl10, Quevedo13TAC, haykin1988digital}.
\end{remark}}

\begin{assumption} \label{asmpt:assumption-stability}
For the least favorable channel power gain $\underline{h} \triangleq \min\{h: h\in \mathbbm{h}\}$, the maximum achievable successful transmission probability satisfies
\[ q(\bar{u},\underline{h}) > 1- \frac{1}{ \lambda_{\mathrm{max}}(A)^2  },   \]
where $\bar{u}$ is the highest available power level: $\bar{u} \triangleq \max\{u: u\in\mathcal{U}\}$ and $A$ is the system matrix in~\eqref{eqn:process-dynamics}.
\end{assumption}
Note that since both $\mathbbm{h}$ and $\mathcal{U}$ are compact, $\underline{h}$ and $\bar{u}$ always exist.
\begin{remark}
Assumption~\ref{asmpt:assumption-stability} provides a sufficient
condition under which the expected estimation error covariance
is bounded, even for the least favorable channel power gain.
Similar assumptions were also adopted in many works, such as~\cite{Quevedo13TAC,nourian2014optimal},
for guaranteeing the stability of the Kalman filtering subject to
random packet losses.
\end{remark}

\subsection{Remote Estimator}
At the base station side, each time a remote estimator generates an estimate based on what it has received from the sensor.
In many applications, the remote estimator is
powered by an external source or is connected with an energy-abundant controller/actuator, thus having sufficient communication energy
in contrast to the energy-constrained sensor. This energy asymmetry allows us to assume that the estimator can send  messages back to the sensor. The content of feedback messages are separatively defined under different system implementations, the details of which will be discussed later in Section~\ref{sec:practical-implementation}.
Denote by $\mathcal{O}^-_k$ the observation
obtained by the remote estimator up to \emph{before}
the communication at time $k$, i.e.,
\begin{equation*}
  \mathcal{O}^-_k\triangleq\{\gamma_1{x}_1,...,\gamma_{k-1}{x}_{k-1}\}\cup
  \{\gamma_{1},\ldots,\gamma_{k-1}\}\cup \{h_{1},\ldots,h_k\}.
\end{equation*}
Similarly, denote by $\mathcal{O}^+_k$ the observation
obtained by the remote estimator up to \emph{after}
the communication at time $k$, where
\begin{equation*}
  \mathcal{O}^+_k\triangleq \mathcal{O}^-_k
  \cup  \{\gamma_{k}, \gamma_k{x}_{k}\}.
\end{equation*}

\section{Problem Definition} \label{sec:remote-state-estim}

We take into account both the estimation quality at the remote estimator and the transmission energy consumed by the sensor.
To this purpose, joint design of the transmission power controller
and the remote estimator is desired.
Measurement realizations, communication indicators, and channel gains are adopted to manage the usage of transmission power:
\begin{equation}\label{def:trans-power-controller}
u_{k}=f_{k}\big(  x_{1:k},h_{1:k},\gamma_{1:k-1}\big).
\end{equation}
To produce $u_k$, the remote estimator
can work under two different implementations, regarding the computational
capacity of the sensor node: the remote estimator
decides an intermediate function
\begin{equation}\label{eqn:lk}
l_k(\cdot)=f_{k}(\mathcal{O}^-_{k}, \cdot )
\end{equation} and feeds $l_k$ back to the
sensor and then the sensor evaluates
$u_k$ by $u_k=l_k(x_{1:k})$; or, equivalently,
the estimator directly feeds $\gamma_{k}$'s back to the sensor, and the sensor is in charge of deciding $f_k$.
Given the transmission power controller, the remote estimator generates an estimate as a function of
what it has received from the sensor, i.e.,{
\begin{equation}\label{def:remote-estimator}
\tilde x_k\triangleq g_k(\mathcal{O}^+_k).
\end{equation}  }
We emphasize that $\tilde x_k$ also depends on $f_k$
since $f_k$ statistically affects the arrival of the data.
The average remote estimation
quality over an infinite time horizon is quantified by
\begin{equation}\label{def:J_E}
\hspace{-2mm}\mathpzc {E}(\mathbf{f},\mathbf{g})\hspace{-0mm}\triangleq
\hspace{-0mm}\mathbb{E}_{\mathbf{f},\mathbf{g}}\left[
\mathop {\rm limsup}_{T\to \infty}\frac{1}{T}
\sum_{j=1}^{T}\|x_k-\tilde x_k\|^2\right];
\end{equation}
correspondingly, the average transmission power cost, denoted as
$\mathpzc{W}(\mathbf{f},\mathbf{g})$, is given by
\begin{equation}\label{def:J_P}
\mathpzc{W}(\mathbf{f})
\triangleq \mathbb{E}_{\mathbf{f}}\left[
\mathop {\rm limsup}_{T\to \infty}
\frac{1}{T}
\sum_{k=1}^{T}u_k\right],
\end{equation}
where $\mathbf{f}\triangleq \{f_1,\ldots,f_k,\ldots\}$
and $\mathbf{g}\triangleq \{g_1,\ldots,g_t,\ldots\}$.
The arguments $\mathbf{f}$ and $\mathbf{g}$ indicate that
the quantity of~\eqref{def:J_E} depends on them.
This is also the case in~\eqref{def:J_P}.
Note that in~\eqref{def:J_E}~and~\eqref{def:J_P} the expectations are
taken with respect to the randomness of the system and the transmission outcomes for given $\mathbf{f}$ and $\mathbf{g}$.
For the remote state estimation system, we naturally wonder how to find a jointly optimal
transmission power controller $f_k^*$ and remote state estimator $g_k^*$ satisfying
\begin{equation} \label{eqn:hiddenproblem}
\mathrm{minimize}_{\mathbf{f},\mathbf{g}}
\left[  \mathpzc{E}(\mathbf{f},\mathbf{g})+
\alpha\,\mathpzc{W}(\mathbf{f})\right],
\end{equation}
where the constant
$\alpha$ can be interpreted as a Lagrange multiplier.
{
We should remark that~\eqref{eqn:hiddenproblem} is difficult to solve due to the nonclassical information structure~\cite{yuksel2013stochastic}.
What is more,~\eqref{eqn:hiddenproblem} has an average cost criterion that depends only on the limiting behavior of $\mathbf{f}$ and $\mathbf{g}$, adding additional analysis difficulty.

\subsection{Belief-State Markov Decision Process}

To find a solution to the optimization problem~\eqref{eqn:hiddenproblem}, we first observe from~\eqref{def:J_P} that $\mathpzc{W}(\mathbf{f})$ does not depend on
${\bf g}$, thus leading to an
insight into the structure of $g_k^*$---Lemma~\ref{lemma:opt-estimator}, the proof of which  follows from optimal filtering theory: the conditional mean is the
minimum-variance estimate~\cite{anderson79}.
Similar results can be seen in~\cite{Lispa11TAC,nayyar2013optimal,molin2014optimal}.
\begin{lemma}\label{lemma:opt-estimator}
For any given transmission power controller $f_k$, the
optimal remote estimator $g_k^\ast$ is
the MMSE estimator
\begin{equation}\label{eqn:mmse-x-k}
\hat x_k\triangleq g_k^\ast(\mathcal{O}^+_k)=
\mathbb{E}_{f_{1:k}}[x_k|\mathcal{O}^+_k].
\end{equation}
\end{lemma}
Problem~\eqref{eqn:hiddenproblem} still remains hard since
$g_k^\ast$
depends on the choice of $f_{1:k}$. To address this issue, by
viewing the problem
from the perspective of a decision maker holding the common information~\cite{nayyar2011sequential}, we formulate~\eqref{eqn:hiddenproblem}
as a POMDP~\cite{Cassandra1998POMDP} at the decision maker's side. Following the conventional treatment of the POMDP, we are allowed to equivalently study its belief-state MDP~\cite{Cassandra1998POMDP}. For technical reasons,
 we pose two moderate constraints on the action space.
We will  present the formal belief-state MDP model and remark
that the resulting gap between the formulated belief-state
MDP and~\eqref{eqn:hiddenproblem} is negligible (see Remark~\ref{remark:gapComments}).
Before doing so, a few definitions and notations are needed.}
Define innovation $e_k$ as
\begin{equation} \label{eqn:Innovation}
e_k\triangleq x_k-
A ^{k-\tau(k)}x_{\tau(k)}
\end{equation}
with $e_k$ taking values in $\mathbb{R}^n$ and $\tau_k$ being the most recent time the remote estimator received data before time $k$ as
\begin{equation}
\tau(k)\triangleq \max_{1\leqslant t \leqslant k-1}\{t :
\gamma_t=1\}.   \label{eqn:tauk}
\end{equation}
Let $\hat e_k\triangleq\mathbb{E}_{f_{1:k}}[e_k|\mathcal{O}^+_k].$
Since $\tau(k), x_{\tau(k)}\in\mathcal{O}^+_{k-1}$, the equality
\begin{equation}\label{eqn:xk-ek}
e_k-\hat e_k=x_k-\hat x_k
\end{equation}
holds for all $k\in\mathbb{N}$.
In other words, $e_k$ can be treated as $x_k$ offset
by a variable that is measurable to $\mathcal{O}^+_{k-1}$.
We define the belief state on $e_k$. From~\eqref{eqn:xk-ek}, the belief state on $x_k$ can be equally defined.
Here we use $e_k$ instead of $x_k$ for ease of presentation.
\begin{definition}
Before the transmission at time $k$, the
belief state $\theta_k(\cdot):\mathbb{R}^n\mapsto\mathbb{R}_+$ is defined
as $\theta_k(e)\triangleq p(e_k; e|f_{1:k}, \mathcal{O}_{k-1}^+)$.
\end{definition}
We also need some definitions related to a partition of a set.
\begin{definition}
A collection $\Delta$ of sets is a partition of a set $\mathcal X$ if the following conditions are satisfied:
\begin{enumerate}
  \item[(i).] $\emptyset \not\in \Delta$.
  \item[(ii).] $\cup_{\mcal B\in\Delta}  \mcal B =  \mcal X$.
  \item[(iii).] If $\mcal B_1,\mcal B_2\in \Delta$ and $\mcal B_1\not=\mcal B_2$, then $\mcal B_1 \cap \mcal B_2 = \emptyset$.
\end{enumerate}
An element of $\Delta$ is also called a cell of $\Delta$.
If $\mcal X \subset\mathbb{R}^n$, we define the size of $\Delta$ as $$|\Delta| \triangleq \sup_{\delta_j,x,y}\{\|x-y\|: x,y\in\delta_j,
\delta_j\in\Delta\}.$$
\end{definition}
\begin{definition}
For two partitions, denoted as $\Delta_1$ and $\Delta_2$, of a set $\mathcal{X}$, $\Delta_1$ is called a refinement of $\Delta_2$ if every cell of $\Delta_1$ is a subset of some cell of $\Delta_2$. Formally it is written as $\Delta_1 \preceq \Delta_2$.
\end{definition}
{One can verify that the relation $\preceq$ is a partial order, and the set of partitions together with this relation form a lattice. We denote the infimum of partitions $\Delta_1$ and $\Delta_2$ as $\Delta_1\wedge\Delta_2$.}



Now we are able to mathematically describe the belief-state MDP  by
a quintuplet
$(\mathbb{N},\mathcal{S},\mathcal{A}, \mathpzc P, \mathpzc C)$. Each
item in the tuple is elaborated as follows. We provide background knowledge of the mathematical notions
used below in the Appendix A.
\begin{enumerate}
\item[(\rmnum{1}).]  The set of decision epochs
is $\mathbb{N}$.

\item[(\rmnum{2}).] State space $\mathcal{S}=\Theta\times \mathbbm{h}$: $\Theta$ is the set of beliefs
over $\mathbb{R}^n$, i.e., the space of probability measures on $\mathbb{R}^n$. The set $\Theta$ is further constrained as follows. Let $\mu$ be a generic element of $\Theta$. Then $\mu$ is equivalent to the Lebesgue measure\footnote{  Let $\mu_1$ and $\mu_2$ be measures on the same measurable space.  Then $\mu_1$ is said to be equivalent to $\mu_2$ if for any Borel subset~$\mcal B$, $\mu_2(\mcal B) = 0 \Leftrightarrow \mu_1(\mcal B)=0$.},  and $\mu$ has the finite second moment, i.e., $\int_{\mathbb{R}^n} \|e\|^2 \mathrm{d} \mu(e) < \infty$.
No generality has been lost by the above two constraints; see Remark~\ref{remark-MDPModel1}.
Let $\theta(e) = \frac{\mathrm{d}\mu(e)}{\mathrm{d}\mathfrak{L}(e)}$ be the Radon--Nikodym derivative~\cite{durrett2010probability}. Note that $\theta(e)$ is uniquely defined up to a $\mathfrak{L}-$null set (i.e., a set having Lebesgue measure zero). We thus use $\mu$ and $\theta(e)$ interchangeably to represent a probability measure on $\mathbb{R}^n$, and we do not distinguish between any two functions $\theta(e)$ and $\theta'(e)$ with $\mathfrak{L}(\{e: \theta(e)-\theta(e)'\not= 0\}) = 0$ by convention.
We assume that $\Theta$ is endowed with the topology of weak convergence. Denote by $s\triangleq (\mu,h)$ a generic element of $\mathcal{S}$.
Let $\mathbbm{d}_P(\cdot,\cdot)$ denote the Prohorov metric~\cite{billingsley2013convergence} on $\Theta$. We define the metric on $\mathcal{S}$ as $\mathbbm{d}_s((\mu_1,h_1),(\mu_2,h_2))=\max
\{\mathbbm{d}_P(\mu_1,\mu_2),|h_1-h_2|\}.$
\item[(\rmnum{3}).] Action space $\mathcal{A}$ is the
set of all functions that have the following structure:
\begin{align} \label{eqn:StrucActionSatur}
a(e)=&\left\{
        \begin{array}{ll}
            \bar{u}, & \text{if $\|e\| > L$}, \\
            a'(e), & \text{otherwise},
        \end{array}
    \right.
\end{align}
where $a'\in\mathcal{A}': \mathcal{E} \mapsto \mathcal U$ with $\mathcal{E} \triangleq \{e\in\mathbb{R}^n: \|e\| \leq L\}$. The space $\mathcal{A}'$ is further defined as follows.
Let $a'\in\mathcal{A}'$ be a generic element, then there exists a finite partition $\Delta_{a'}$ of $\mathcal{E}$ such that each cell of $\Delta_{a'}$ is a $\mathfrak{L}-$continuity set\footnote{A Borel subset $\mcal B$ is said to be a $\mu-$continuity set if $\mu(\partial \mcal B) = 0$, where $\partial \mcal B$ is the boundary set of $\mcal B$.} and on each cell $a'(e)$ is  Lipschitz continuous with Lipschitz constant uniformly bounded by $M$. It is further assumed that $\overline{\Delta} = \wedge_{a'\in\mathcal{A}'} \Delta_{a'}$ is a finite partition of $\mathcal{E}$. Since both $L$ and $M$ can be arbitrarily large and $|\overline{\Delta}|$ can be arbitrarily small, the structure constraints pose little limitation; see Remark~\ref{remark-MDPModel2} why we consider such an action space.
We consider the Skorohod distance defined in~\eqref{eqn:SkorohodDistance}. By convention, we do not distinguish two functions in $\mathcal{A}$ that have zero distance and we consider the space of the resulting equivalence classes.
Note that the argument of the function $a(\cdot)$ is the innovation $e_k$ defined in~\eqref{eqn:Innovation}, and by the definition of $e_k$, one obtains that $a_k(e) = l_k(e+A ^{k-\tau(k)}x_{\tau(k)})$.

\item[(\rmnum{4}).] The function $\mathpzc  P(\theta', h'|\theta, h,a):\mathcal{S}\times \mathcal{A}
\times \mathcal{S}$ defines the conditional state transition
probability.
To be precise,
\begin{align*}
&\mathpzc P(\theta', h'|\theta, h, a)\\
&\triangleq p(\theta_{k+1}, h_{k+1};\theta', h'|\theta_k=\theta, h_k=h, a_{k}=a)\\
&=\left\{\begin{array}{l}
\pi(h'|h)\left( 1-\varphi(\theta, h,a)  \right)
\delta_{\phi(\theta,h,a,0)}(\theta'),\text{~if\:}
\theta'=\phi(\theta,h,a,0),\\
\pi(h'|h)\varphi(\theta, h,a)
\delta_{\phi(\theta,h,a,1)}(\theta'), ~\text{if\:}
 \theta'=\phi(\theta,h,a,1),  \\
0, \qquad \qquad \qquad \qquad \text{otherwise},
\end{array}\right.
\end{align*}
where $\varphi(\theta, h,a) \triangleq \int_{\mathbb{R}^n}q(a(e),h)\theta(e)\mathrm{d}e $, $\delta_{\theta_\ast}(\theta)$ denotes a degenerate distribution
at $\theta_\ast$ over $\Theta$ with $\int_{\Theta}\delta_{\theta_\ast}(\theta){\rm d}\theta=1$, and
\begin{align*}
&\phi(\theta, h,a, \gamma)\\
&\triangleq \left\{
\begin{array}{lll}
\frac{1}{\mathrm{det}(A)}\theta^{+}_{\theta,h,a}(A^{-1}e) * \mathscr{N}_{0,W}(e),
& \hbox{if~} \gamma=0,\\
\mathscr{N}_{0,W}(e), &
\hbox{if~} \gamma=1,
\end{array}\right.  \addtag  \label{eqn:BliefStateUpdate}
\end{align*}
where  $\theta^{+}_{\theta,h,a}(e) \triangleq \frac{(1-q(a(e),h))\theta(e)}{
1-\varphi(\theta, h,a)}$ is interpreted as the post-transmission belief when the transmission fails, and $\mathscr{N}_{0,W}(e)$ is the multivariate Gaussian distribution with mean $0$ and covariance $W$.

\item[(\rmnum{5}).] The function $\mathpzc C(\theta,h,a):\mathcal{S}\times \mathcal{A} \to \mathbb{R}_+
$ is the cost function when performing $a\in \mathcal{A}$ for
$\theta\in\Theta$ and $h\in\mathbbm{h}$ at time $k$, which is given by
\begin{equation}\label{eqn:cost-func}
\mathpzc C(\theta, h, a)=\int_{\mathbb{R}^n}\theta(e)c(e, h, a){\rm d} e.
\end{equation}
In~\eqref{eqn:cost-func}, the function
$c(\cdot,\cdot,\cdot): \mathbb{R}^n\times \mathbbm{h}\times \mathcal{A}\mapsto \mathbb{R}_+$ is defined
as $c(e, h, a)=\alpha a(e)+(1-q(a(e),h))\|e-\hat{e}_+\|^2$ with $\hat{e}_+ = \mathbb{E}_{\theta^{+}_{\theta,h,a}}[e] \triangleq  \mathbb{E}[e|e\sim\theta^{+}_{\theta,h,a}]$, where
the communication cost is counted by the first term and
the distortion $\|e-\hat e_+\|^2$ with probability $1-q(a(e),h)$
is counted by the second term.
\end{enumerate}
\begin{remark} \label{remark-MDPModel1}
The initial belief $\theta_1(e) = 1/\mathrm{det}(A)\mu_{x_0}(A^{-1}e) * \mathscr{N}_{0,W}(e)$ is equivalent to the Lebesgue measure. The belief evolution in~\eqref{eqn:BliefStateUpdate} gives that, whatever policy is used, $\theta_k$ is equivalent to the Lebesgue measure for
$k\geq 2$. Also, note that if there exists a channel gain $h\in\mathbbm{h}$ such that $q(\bar{u},h)<1$ and if $\theta$ has infinite second moment, then $\mathpzc C(\theta, h, a) = \infty$ for any action $a$. Thus, to solve~\eqref{eqn:hiddenproblem}, without any performance loss, we can restrict beliefs into the state space $\Theta$.
\end{remark}

\begin{remark} \label{remark-MDPModel2}
The action $a(e)\in\mathcal{A}$ is allowed to have a $\mathfrak{L}-$null set of discontinuity points. The assumption that on each cell of a partition, $a(e)$ is a Lipschitz function is a technical requirement for Theorem~\ref{theorem:ExistenceStationary}. The intuition is that given $\theta_k$, except for $\mathfrak{L}-$null set of points, the difference between the power used for $e_k$ and $e_k'$ is at most proportional to the distance between $e_k$ and $e_k'$. The saturation structure in~\eqref{eqn:StrucActionSatur}, i.e., $a(e)=\bar{u}$ if $\|e\| > L$ is also a  technical requirement for Theorem~\ref{theorem:ExistenceStationary}. Intuitively, this ensures that, when the transmission fails, the second moment of the post-transmission belief  $\theta^{+}_{\theta,h,a}(e)$ is bounded by a function of the second moment of $\theta(e)$. The saturation assumption can also be found in~\cite{gatsis2014optimal}.
\end{remark}

An admissible $k-$history for this MDP is defined as $\mathbf{h}_k \triangleq \{\theta_1, h_1,a_1,\ldots,\theta_{k-1},h_{k-1},a_{k-1},\theta_k,h_k\}$. Let $\mathcal{H}_k$ denote the class of all the admissible $k-$history $\mathbf{h}_k$.
A generic policy $\mathbf{d}$ for
$(\mathbb{N},\mathcal{S},\mathcal{A},\mathpzc P,\mathpzc C)$
is a sequence of
decision rules $\{d_{k}\}_{k\in \mathbb{N}}$, with each
$d_k:\mathcal{H}_k\to\mathcal{A}$. In general, $d_k$ may be a stochastic mapping. Let $\mathcal{D}$ denote the whole class of $\mathbf{d}$. In some cases, we may write $\mathbf{d}$ as $\mathbf{d}(d_{k})$ to explicitly point out the decision rules used at each stage. We focus on the following problem.
\begin{problem} \label{problem:2}
Find optimal policy for the  MDP problem:
\begin{align}
&\mathop {\rm minimize}_{\mathbf{d}\in\mathcal{D}} \mathpzc{J}(\mathbf{d},\theta,h), \quad \forall\, (\theta,h) \in \mathcal{S}, \label{eqn:ACMDP}
\end{align}
where
\begin{align*}
&\mathpzc{J}(\mathbf{d},\theta,h) \triangleq
\mathop {\rm limsup}_{T\to \infty}\frac{1}{T}
\sum_{k=1}^{T} \mathbb{E}^{\theta,h}_{\mathbf{d}}\left[
\mathpzc C(\theta_k, h_k, a_k) \right]
\end{align*}
is the average cost with the initial state $(\theta,h)$ and the policy~$\mathbf{d}$.
\end{problem}
{\begin{remark} \label{remark:gapComments}
The gap between Problem~\ref{problem:2} and the optimization problem~\eqref{eqn:hiddenproblem} arises from the structure assumptions for the action space. These structure constraints, however, are moderate, since the saturation level $L$ and the uniform Lipschitz constant $M$ can be arbitrarily large and the size of  $|\overline{\Delta}|$ can be arbitrarily small.
\end{remark}}

\begin{center}
\setlength{\unitlength}{1.5mm}
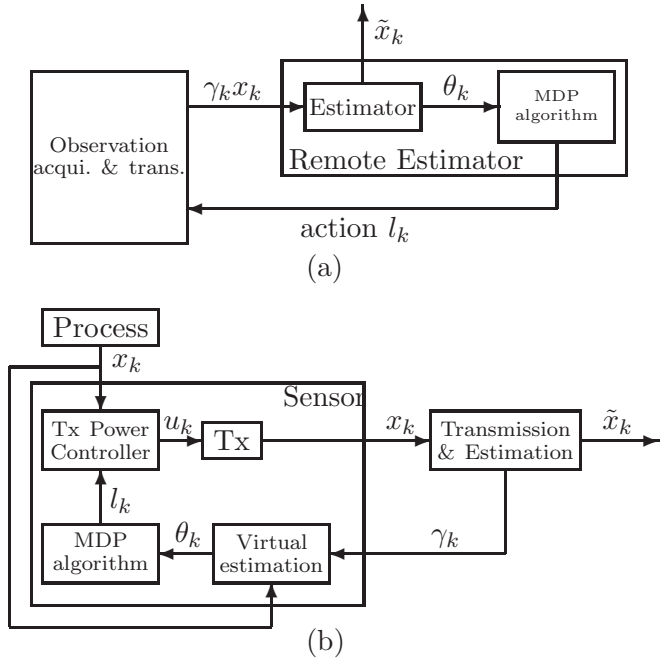
\begin{figure}[t!]
\thicklines
\centering
\begin{picture}
(56,59)(-3,-13)
\thicklines

\put(0,24){\framebox(13.5,15){${\rm Observation\atop acqui.~\&~trans.}$}}

\put(13.5,36){\vector(1,0){10.5}}
\put (15,37){$ \gamma_kx_k$}

{\footnotesize
\put(24,34){\framebox(10,4){${\rm Estimator}$}}
\put(41,33){\framebox(10,6){${\rm MDP \atop algorithm}$}}
\put(22,30){\framebox(30,10)}
}
\put (22.5,30.5){$\rm Remote~Estimator$}
\put(29,38){\vector(0,1){7}}
\put (30,42){$\tilde x_k$}
\put(34,36){\vector(1,0){7}}
\put (36,37){$\theta_k$}

\put(46,33){\line(0,-1){6}}
\put(46,27){\vector(-1,0){32.5}}
\put (23.5,24.5){${\rm action}~l_k$}
\put (24,21){$\rm (a)$}

\put(1,15){\framebox(10,3){${\rm Process}$}}
\put(7,13){$x_k$}
\put(6,15){\vector(0,-1){6}}
\put(1,4){\framebox(10,5){${\rm Tx~Power\atop Controller}$}}
\put(15,5){\framebox(5,3){${\rm Tx}$}}
\put(11,6.5){\vector(1,0){4}}
\put(11.5,7.5){$u_k$}

\put(1,-6){\framebox(10,5){${\rm MDP \atop algorithm}$}}
\put(6,-1){\vector(0,1){5}}
\put (6.8,0.5){$l_k$}

\put(16,-6){\framebox(10,5){${\rm Virtual \atop estimation}$}}
\put(16,-3.5){\vector(-1,0){5}}
\put (12.5,-2.5){$\theta_k$}

\put(20,6.5){\vector(1,0){15}}
\put (31,7.5){$x_k$}

\put(35,4){\framebox(13,5){${\rm Transmission}\atop \&~{\rm Estimation}$}}
\put(48,6.5){\vector(1,0){7}}
\put (50,7.5){$\tilde x_k$}

\put(41.5,4){\line(0,-1){7.5}}
\put(41.5,-3.5){\vector(-1,0){15.5}}
\put (35,-2.5){$\gamma_k$}

\put(0, -8){\framebox(29,19.5)}
\put (22,9.5){$\rm Sensor$}
\put(6,13){\line(-1,0){8}}
\put(-2,13){\line(0,-1){23}}
\put(-2,-10){\line(1,0){23}}
\put(21,-10){\vector(0,1){4}}

\put (24,-12){$\rm (b)$}

\end{picture}
\caption{Implementation of the system. The block ``Observation acqui. $\&$ trans." in (a) corresponds to the blue-dashed rectangle in Figure~\ref{fig:system} and the block ``Transmission $\&$ estimation" the red-dashed rectangle. In (a), the MDP algorithm is run at the remote estimator and the action $l_k$ is fed back to the sensor. While in (b), the MDP algorithm is run at the sensor node and $\gamma_k$ is fed back by the remote estimator.
}
\label{implementation-model}
\end{figure}
\end{center}
\vspace{-0.5cm}
{
\subsection{Practical Implementation}\label{sec:practical-implementation}
Here we discuss about the implementation of the system, which is illustrated in Figure~\ref{implementation-model}. Depending on the computational capacity of the senor node, the system we study can work either as in (a) or in (b).}
The main difference between the systems in (a) and (b) is where the MDP algorithm is run. The computational capacity required for the sensor node as well as the content of feedback messages are correspondingly different. In (a), the MDP algorithm is run at the remote estimator and the action $l_k$ is fed back to the sensor.
In practice, for a generic $l_k$, only an approximate version (e.g., lookup tables) can be transmitted due to bandwidth limitation.
{An accurate feedback of $l_k$ is possible if $l_k$ has a special structure. For example, if $a_k(e)$ (recall that that $a_k(e) = l_k(e+A ^{k-\tau(k)}x_{\tau(k)})$) is a monotonic step function taking values in a finite set, only those points, where $a_k$ jumps,  are needed to represent $l_k$ (note that $A ^{k-\tau(k)}x_{\tau(k)}$ is available at the sensor node).}
Since the function $l_k$ is directly fed back to the sensor,
the only computational task carried out by the sensor is computing $l_k(x_k)$. When the sensor node is capable of running the MDP algorithm locally, the system can be implemented as illustrated in (b). In this case, only $\gamma_k$ (a binary variable) is fed back. Note that when $\gamma_k$ is fed back, the sensor knows exactly the information available at the remote estimator. It can run a virtual estimator locally that has the same behavior as the remote estimator.


\section{Optimal Deterministic Stationary Policy: Existence} \label{sec:ExistOptPolicy}

The definition of the policy $\mathbf{d}$ in the above section allows the dependence of $d_k$ on the full $k-$history, $\mathbf{h}_k$. Fortunately, with the aid of the results of average cost MDPs~\cite{schal1993average, feinberg2012average, hernandez2012discrete}, we prove that there exists a deterministic stationary policy that is optimal to Problem~\ref{problem:2}. Before showing the main theorem, we introduce some notations.

We define the class of deterministic and stationary policies $\mathcal{D}_{\rm ds}$ as follows:
$\mathbf{d}(d_{k}) \in\mathcal{D}_{\rm ds}$ if and only if there exists a Borel measurable function $d: \mathcal{S} \mapsto \mathcal{A}$ such that $\forall i$,
\begin{align*}
d_k(\mathcal{H}_{k-1},a_{k-1},\theta_k =\theta, h_k=h) = d( \theta, h).
\end{align*}
Since the decision rules $d_k$'s are identical (equal $d$) along the time horizon for a stationary policy $\mathbf{d}(\{d_{k}\}_{k\in \mathbb{N}}) \in\mathcal{D}_{\rm ds}$, we write it as $\mathbf{d}(d)$ for the ease of notation.

\begin{theorem}  \label{theorem:ExistenceStationary}
There exists a deterministic and stationary policy $\mathbf{d}^{*}(d)\in \mathcal{D}_{\rm ds}$  such that
\begin{align*}
\mathpzc{J}(\mathbf{d}^{*}(d),\theta,h)  \leq \mathpzc{J}(\mathbf{d},\theta,h) \quad \forall\, (\theta,h) \in \mathcal{S}, \mathbf{d}\in \mathcal{D},
\end{align*}
Moreover,
\begin{align} \label{eqn:OptimalPolicy}
\mathbf{d}^{*}(d) = \mathop{\arg\min}_{\mathbf{d} \in \mathcal{D}_{\rm ds}} \{\mathpzc C_d(\theta, h) - \rho^*(\theta,h) + \mathbb{E}_{\bf d}[\mathpzc Q^*(\theta', h')|{\theta,h} \},
\end{align}
and
\begin{equation*}
\mathpzc{J}(\mathbf{d}^{*}(d),\theta,h)=\rho^*(\theta,h),
\end{equation*}
where the functions $\mathpzc Q^*: \mathcal{S} \mapsto \mathbb{R}$ and $\rho^*: \mathcal \mapsto \mathbb R$ satisfy
\[\mathpzc Q^*(\theta, h) = \min_{\mathbf{d} \in \mathcal{D}_{\rm ds}} \{
\mathpzc C_d(\theta, h) - \rho^*(\theta,h) + \mathbb{E}_{\bf d}[\mathpzc Q^*(\theta', h')|{\theta,h}]\}\]
with $\mathpzc C_d(\theta, h) \triangleq \mathpzc C(\theta, h, d(\theta, h))$ and
$\mathbb{E}_{\bf d}[\mathpzc Q^*(\theta', h')|{\theta,h}] \triangleq \int_{\mathcal{S}} \mathpzc Q^*(\theta', h') \mathpzc P(\theta',h'|\theta,h, d(\theta,h)) \mathrm{d}(\theta',h')$.
\end{theorem}

The above theorem says that the optimal power transmission policy exists and is deterministic and stationary, i.e., the power used at the sensor node $u_k$ only depends on $(\theta_k,h_k)$ and $e_k$.
Since the belief state $\theta_k$ can be updated recursively as in~\eqref{eqn:BliefStateUpdate}, this property facilitates the related performance analysis. In principle, the optimal deterministic and stationary policy to an average cost MDP can be obtained by the value iteration algorithm~\cite{hernandez2012discrete}. However, it is not computationally tractable to solve~\eqref{eqn:OptimalPolicy}, since neither the state space nor the action space is finite. An approximate algorithm is needed to obtain an suboptimal solution~\cite{Cassandra1998POMDP}, which is out of the scope of this paper. Nevertheless, Theorem~\ref{theorem:ExistenceStationary} provides a qualitative characteristic of the optimal transmission power control rule.


Define
\begin{align*}
\mathpzc{J}_\beta(\mathbf{d},\theta,h) \triangleq
\mathbb{E}_{\theta,h}^{\mathbf{d}}\left[\sum_{k=1}^{\infty} \beta^k\mathpzc C(\theta_k, h_k, a_k)\right]
\end{align*}
as the expected total discounted cost with the discount factor $0<\beta<1$. Let $\upsilon_\beta(\theta,h) \triangleq  \inf_{\mathbf{d}\in\mathcal{D}}\mathpzc{J}_\beta(\mathbf{d},\theta,h)$ be the least cost associated with the initial state $(\theta,h)$, and let $m_\beta = \inf_{(\theta,h)\in\mathcal{S}} \upsilon_\beta(\theta,h).$

By Theorem 3.8 in~\cite{schal1993average}, in order to prove Theorem~\ref{theorem:ExistenceStationary}, it is sufficient to verify the following conditions.

\begin{description}
  \item[C1] {\em (State Space)} The state space $\mathcal{S}$ is locally compact with countable base.
  \item[C2] {\em (Regularity)} Let $\mathpzc{M}$ be a mapping assigning to each $s\in\mathcal{S}$ the nonempty available action space $\mathcal{A}(s)$. Then for each $s\in\mathcal{S}$, $\mathcal{A}(s)$ is compact, and $\mathpzc{M}$ is upper semicontinuous.
  \item[C3] {\em (Transition Kernel)} The state transition kernel $\mathpzc P(\cdot|s,a)$ is weakly continuous\footnote{We say $\mathpzc P(\cdot|s,a)$ is is weakly continuous if as $i\to\infty$,
      \begin{align*}
      \int_{\mathcal{S}}b(s')\mathpzc P(\mathrm{d} s'|s_i,a_i)\to \int_{\mathcal{S}}b(s')\mathpzc P(\mathrm{d} s'|s,a)
      \end{align*}
  for any sequence $\{(s_i,a_i), i\geq 1\}$ converging to $(s,a)$ with $s_i,s\in\mathcal{S}$ and $a_i,a\in\mathcal{A}$, and for any bounded and continuous function $b:\mathcal{S}\mapsto \mathbb{R}$.}.
  \item[C4] {\em (Cost Function)} The one stage cost function $\mathpzc C(s,a)$ is lower semicontinuous.
  \item[C5] {\em (Relative Discounted Value Function)}  There holds \begin{align}\label{eqn:ExistConditionB0}
\sup_{0<\beta<1}[\upsilon_\beta(\theta,h)-m_\beta] < \infty,\: \forall (\theta,h)\in\mathcal{S}.
\end{align}
\end{description}

We now verify each of the above conditions for the considered problem, by which we establish the proof of Theorem~\ref{theorem:ExistenceStationary}.

\subsection{State Space Condition C1}
We prove that both $\mathcal{S}$ and $\mathcal{A}$ are Borel subsets of Polish spaces (i.e., separable completely metrizable topological spaces) instead. Then as pointed out in~\cite{feinberg2012average}, by Arsenin--Kunugui Thoerem, the condition C1 holds.

To show that $\mathcal{S}$ is a Borel subset of a Polish space,
by the well known results about the product topology~\cite{MInyan2010Topology}, it suffices to prove that $\Theta$ and $\mathbbm{h}$ are Borel subsets of Polish spaces.
Since $\mathbbm{h}$ is a compact subset of $\mathbb{R}$, we only need to prove $\Theta$ is a Borel subset of a Polish space. Let $\mathcal {M}(\mathbb{R}^n)$ be the space of probability measures on $\mathbb{R}^n$ endowed with the topology of weak convergence. It is well known that $\mathcal{M}(\mathbb{R}^n)$ is a Polish space~\cite{billingsley2013convergence}. Let $ \mathcal{M}_2(\mathbb{R}^n) \subseteq \mathcal{M}(\mathbb{R}^n)$ be the set of probability measures with finite second moment, and $\mathcal{M}_{\rm e}(\mathbb{R}^n) \subseteq \mathcal{M}(\mathbb{R}^n)$ be the set of probability measures equivalent to $\mathfrak{L}$. By Theorem 3.5 in~\cite{lange1973borel}, $\mathcal{M}_{\rm e}(\mathbb{R}^n)$ is a Borel set. We then show that $\mathcal{M}_2(\mathbb{R}^n)$ is closed. Suppose $\{\mu_{i, i\in\mathbb{N}}\} \in \mathcal{M}_2(\mathbb{R}^n) $ and $\mu_i \overset{w}{\to} \mu$. Since $\mathcal{M}(\mathbb{R}^n)$ is complete, $\mu\in\mathcal{M}(\mathbb{R}^n)$, and using the fact that norms are continuous, by Theorem 1.1 in~\cite{feinberg2014fatou},
\[\int_{\mathbb{R}^n} \|e\|^2 \mu(\mathrm{d} e) \leq \liminf_{i\to\infty}  \int_{\mathbb{R}^n} \|e\|^2 \mu_i(\mathrm{d} e) < \infty.   \]
Then $\mu\in\mathcal{M}_2(\mathbb{R}^n)$, implying that $\mathcal{M}_2(\mathbb{R}^n)$ is closed. Since $\Theta = \mathcal{M}_2(\mathbb{R}^n) \cap \mathcal{M}_{\rm e}(\mathbb{R}^n)$, $\Theta$ is a Borel subset of $\mathcal{M}(\mathbb{R}^n)$.
The state space $\mathcal{S}$ thus is a Borel subset of a Polish space.

Now we shall show that $\mathcal A$ is a Borel subset of
a Polish space.
Since a bounded function can be approximated by simple functions uniformly~\cite{rudin1964principles}, the space $\mathcal{A}'$ is a subset of the general Skorohod space defined on $\mathcal{E}$ (see Appendix A), i.e., $\mathcal{A}' \subseteq \mathcal{D}(\mathcal{E})$. We first prove that the closure of $\mathcal{A}'$, denoted as $\mathrm{cl}(\mathcal{A}')$, is a compact set by Theorem 3.11 in~\cite{straf1972}. Since a generic $a\in\mathcal{A}'$ maps from $\mathcal{E}$ to $[0,\bar{u}]$, the condition 3.37 is obviously satisfied. Note that the condition 3.38 is equivalent to
\begin{align} \label{eqn:SkorohodFuncwLim}
\lim_{\Delta}\sup_{a\in\mathcal{A}'}\mathbbm{w}(a,\Delta) \to 0.
\end{align}
By the definition of $\overline{\Delta}$, \emph{all} the functions in $\mathcal{A}'$ are Lipschitz continuous with Lipschitz constant uniformly bounded by $M$ on each cell of $\overline{\Delta}$. Thus, for $\Delta \preceq\overline{\Delta}$,
\begin{align*}
\sup_{a\in\mathcal{A}'}\mathbbm{w}(a,\Delta) \leq M|\Delta|,
\end{align*}
which yields~\eqref{eqn:SkorohodFuncwLim}. We then show that $\mathrm{cl}(\mathcal{A}') = \mathcal{A}'$. Suppose that $a_i\in\mathcal{A}'$ converges to a limit $a$ in the Skorohod topology (we write as $a_i\overset{s}{\to}a$), we then show that $a\in\mathcal{A}'$. By the definition of the Skorohod distance $\mathbbm{d}(\cdot,\cdot)$ in~\eqref{eqn:SkorohodDistance}, $a_i\overset{s}{\to}a$ if and only if there exist mappings $\bbpi_i\in \Lambda_t$ such that
\begin{align}  \label{eqn:ExistAction_1}
\lim_i a_i(\bbpi_i x) = a(x)\: \text{uniformly in} \:\mathcal{E}
\end{align}
and $\lim_i \bbpi_i x = x$ uniformly in $\mathcal{E}$.
Since $\lim_i \bbpi_i x = x$ uniformly in $\mathcal{E}$, for any $\epsilon>0$, there exists $i_0$ such that $\|\bbpi_i\|_t < \epsilon$ with $i\geq i_0$. Note that if $\|\bbpi_i\|_t<\epsilon$, $\bbpi_i$ is a bi-Lipschitz
homeomorphism. By the definition of $\mathcal{A}'$, any $a_i\in\mathcal{A}'$ has $\mathfrak{L}-$null set of discontinuity points. Since measure-null sets are preserved by a Lipschitz homeomorphism, by~\eqref{eqn:ExistAction_1}, one obtains that
\begin{align}  \label{eqn:ExistAction_2}
\mathfrak{L}(\text{the set of discontinuity points of $a$}) = 0.
\end{align}
Following the same reasoning for one dimensional Skorohod space $\mathscr{D}[0,1]$ (see e.g., P124 in~\cite{billingsley2013convergence}), one obtains that $a_i\overset{s}{\to}a$ implies that $a_i(x)\to a(x)$ uniformly for all continuity points $x$ of $a$. Since on each cell of $\overline{\Delta}=\{\delta_j\}$, \emph{all} the functions in $\mathcal{A}'$ are Lipschitz continuous, the interior points of $\delta_j$ (write the set as $\delta_j^o$) must be continuity points of $a$. By the fact that if a sequence of Lipschitz functions with Lipschitz constant uniformly bounded by $M$ converge to a limit function, then this limit function is also a Lipschitz function with Lipschitz constant bounded by the same $M$, $a$ is Lipschitz continuous with Lipschitz constant uniformly bounded by $M$ on the interior set of each cell of $\overline{\Delta}$. For a boundary point $x$ of the cells of $\overline{\Delta}$, denote the collection of cells whose boundary contains $x$ as $\delta_x \triangleq \{\delta_j: x\in\partial{\delta_j}\}$. Then one obtains that $a(x)$ must be a limit of $a$ from one cell in $\delta_x$, i.e., there exists $\delta_j\in\delta_x$ such that $\lim_{y\to x, y\in\delta^o_j}a(y) = a(x)$. Now we define a function $a^*$ such that for each $\delta_j \in \overline{\Delta}$, $a^*(x) = a(x)$ if $x\in\delta_j^o$ and $a^*(x)$ are continuous on $\delta_j$. Then one obtains that $\mathbbm{d}(a,a^*)=0$, which implies that $a=a^*$ since $\mathscr{D}(\mathcal{E})$ is a metric space.
Combining~\eqref{eqn:ExistAction_2}, one obtains that $a\in\mathcal{A}'$.
Thus $\mathcal{A}'$ is closed and compact.
Using the fact that  every compact metric space is complete and separable, one obtains that $\mathcal{A}'$ is a Polish space. By the structure relation between $\mathcal{A}$ and $\mathcal{A}'$ in~\eqref{eqn:StrucActionSatur}, the space $\mathcal{A}$ is also a Polish space.

\subsection{Regularity Condition C2}
Since $\mathcal{A}$ is compact and $\mathcal{A}(s) = \mathcal{A}$ for every $s\in\mathcal{S}$, C2 is readily verified.

\subsection{Transition Kernel Condition C3}
Since $\mathcal{S}$ is separable and given $(\theta_k,h_k,a)$, $h_{k+1}$ and $\theta_{k+1}$ are independent, then by Theorem 2.8 in~\cite{billingsley2013convergence}, it suffices to prove that for any $h\in\mathbbm{h}$, as $\theta_i \overset{w}{\to} \theta$ ($\mu_i \overset{w}{\to} \mu$) and $a_i\overset{s}{\to}a$, the followings hold:
\begin{align}
\varphi(\theta_i, h,a_i) &\to \varphi(\theta, h,a)\label{eqn:ExistRateConv}  \\
\text{and}\quad\phi(\theta_i, h,a_i, 0) &\overset{w}{\to} \phi(\theta, h,a, 0). \label{eqn:ExistDistConv}
\end{align}
Since $a_i\overset{s}{\to}a$ implies that $a_i(x)\to a(x)$ uniformly for all the continuity points of $a$ and the set of discontinuity points of $a$ has Lebesgue measure zero, $a_i\overset{s}{\to}a$ implies that $a_i\to a$ $\mathfrak{L}-$a.e.
Noting that $\mu$ is equivalent to $\mathfrak{L}$, $a_i\to a$ $\mu-$a.e. holds, and it follows that $q(a_i,h) \to q(a,h)$ $\mu-$a.e., since $q$ is continuous $\mathfrak{L}-$a.e.
Also, by Lemma~\ref{lemma:WeakConv2Setwise}, $\mu_i\overset{sw}{\to}\mu$. Then by Theorem 2.2 in~\cite{hernandez2000fatou}, one obtains that
\begin{align*}
\liminf_{i\to\infty} \int_{\mathbb{R}^n} q(a_i(e),h) \mu_i(\mathrm{d}e) \geq&  \int_{\mathbb{R}^n} q(a(e),h) \mu(\mathrm{d}e) \\
\text{and}\: \liminf_{i\to\infty} \int_{\mathbb{R}^n} -q(a_i(e),h) \mu_i(\mathrm{d}e) \geq&  -\int_{\mathbb{R}^n} q(a(e),h) \mu(\mathrm{d}e).
\end{align*}
Combing the above two equations, one obtains that $\lim_{i\to\infty} \int_{\mathbb{R}^n} q(a_i(e),h) \mu_i(\mathrm{d}e) \to  \int_{\mathbb{R}^n} q(a(e),h) \mu(\mathrm{d}e)$, i.e., $\varphi(\theta_i, h,a_i) \to \varphi(\theta, h,a)$.

We now prove that the equation~\eqref{eqn:ExistDistConv} holds. Noting that $\theta_i\overset{sw}{\to}\theta$ implies that $\theta_i(e){\to}\theta(e)$ $\mathfrak{L}-$a.e., it thus follows that
\begin{align} \label{eqn:ExistPlusConv}
\theta^{+}_{\theta_i,h,a_i}(e) \to \theta^{+}_{\theta,h,a}(e)
\end{align}
$\mathfrak{L}-$a.e. Note that $\theta^{+}_{\theta_i,h_i,a_i}(e)$ and $\theta^{+}_{\theta,h,a}(e)$ can be viewed as probability density functions of $e$, and for simplicity, we write the corresponding probability measures as $\mu^+_i$ and $\mu^+$, respectively. Then it follows from~\eqref{eqn:ExistPlusConv} that
\begin{align} \label{eqn:ExistPlusConv2}
\mu^+_i \overset{sw}{\to} \mu^+.
\end{align}
Let $b(e)$ be any bounded and continuous function defined on $\mathbb{R}^n$, then
\begin{align*}
&\int_{\mathbb{R}^n} b(e)   \phi(\theta, h,a, 0)(e) \mathrm{d}e \\
& = \int_{\mathbb{R}^n} b(e)   \int_{\mathbb{R}^n} \theta^{+}_{\theta,h,a}(e') \mathscr{N}_{0,W}(e-Ae')\mathrm{d}e'      \mathrm{d}e  \\
& = \int_{\mathbb{R}^n} \theta^{+}_{\theta,h,a}(e')  \int_{\mathbb{R}^n} b(e)   \mathscr{N}_{0,W}(e-Ae')\mathrm{d}e     \mathrm{d}e' \\
& \triangleq \int_{\mathbb{R}^n} \tilde{b}(e')\mu^+(\mathrm{d}e'),
\end{align*}
where $\tilde{b}(e') \triangleq \int_{\mathbb{R}^n} b(e)   \mathscr{N}_{0,W}(e-Ae')\mathrm{d}e$. Noting that $\tilde{b}(e')$ is a bounded function, then by Appendix E of~\cite{hernandez2012discrete} and \eqref{eqn:ExistPlusConv2},
\[\int_{\mathbb{R}^n} \tilde{b}(e')\mu^+_i(\mathrm{d}e') \to \int_{\mathbb{R}^n} \tilde{b}(e')\mu^+(\mathrm{d}e'). \]
The equation~\eqref{eqn:ExistDistConv} thus follows by the Portmanteau Theorem.

\subsection{Cost Function Condition C4}
 We first show that $\mu^+$ also has finite second moment given that $\mu$ has finite second moment.
\begin{align*}
&\int_{\mathbb{R}^n} \|e\|^2 \mathrm{d} \mu^+(e) \\
&= \int_{\mathcal{E}} \|e\|^2 \mathrm{d} \mu^+(e) +\int_{\mathbb{R}^n\backslash \mathcal{E}} \|e\|^2 \mathrm{d} \mu^+(e) \\
&\leq L^2 + \int_{\mathbb{R}^n\backslash \mathcal{E}} \|e\|^2 \mathrm{d} \mu(e), \:\text{for any}\: h,a \\
&<\infty,
\end{align*}
where the first inequality  follows from the structure of $a(e)$ in~\eqref{eqn:StrucActionSatur}.
Since $\theta^{+}_{\theta,h,a}$ has finite second moment, $e\sim\theta^{+}_{\theta,h,a}$ is uniformly integrable. Then by Theorem 3.5 in~\cite{billingsley2013convergence},
\begin{align*}
\hat{e}^i_+ \to \hat{e}_+
\end{align*}
where $\hat{e}^i_+ = \mathbb{E}_{\theta^{+}_{\theta_i,h,a_i}}[e].$
Note that
\begin{align*}
\mathpzc C(\theta, h, a)= &\int_{\mathbb{R}^n}\theta(e)c(e, h, a){\rm d} e. \\
=&\int_{\mathbb{R}^n}\alpha a(e)+(1-q(a(e),h))\|e-\hat{e}_+\|^2 \mathrm{d}\mu(e).
\end{align*}
Since $a_i(e)+(1-q(a_i(e),h_i))\|e_i-\hat{e}^i_+\|^2 \geq 0$, then by Theorem 2.2 in~\cite{hernandez2000fatou}, one obtains that
\begin{align*}
 &\int_{\mathbb{R}^n}\alpha a(e)+(1-q(a(e),h))\|e-\hat{e}_+\|^2 \mathrm{d}\mu(e) \\
 & \leq \liminf_{i\to\infty} \int_{\mathbb{R}^n}\alpha a_i(e)+(1-q(a_i(e),h))\|e-\hat{e}^i_+\|^2\mathrm{d}\mu_i(e),
 \end{align*}
which means that $\mathpzc C(\theta, h, a)$ is lower semicontinuous.

\subsection{Relative Discounted Value Function Condition C5}
Note that by Lemma 5 in~\cite{feinberg2012average}, if
\begin{align} \label{eqn:ExistConditionB1}
\inf_{\mathbf{d},\theta,h}\mathpzc{J}(\mathbf{d},\theta,h) < \infty,
\end{align}
then~\eqref{eqn:ExistConditionB0} can be equivalently written as
\begin{align} \label{eqn:ExistConditionB}
\limsup_{\beta \uparrow 1}[\upsilon_\beta(\theta,h)-m_\beta] < \infty,\: \forall (\theta,h)\in\mathcal{S}.
\end{align}

To verify~\eqref{eqn:ExistConditionB1}, consider a suboptimal policy, denoted by $\mathbf{d}^\diamond$, where at each time instant the maximal transmission power $\bar{u}$ is used.
Given a belief $\theta$, denote by $\mathtt{Var}(\theta)$ the second central moment, i.e.,
\begin{align} \label{eqn:define2Centmoment}
\mathtt{Var}(\theta)=\int_{\mathbb{R}^n}\theta(e)(e-\hat{e})(e-\hat{e})^\top{\rm d} e,
\end{align}
where $\hat{e} = \mathbb{E}[e|e\sim\theta]$ is the mean. Then
for any initial state $(\theta,h)\in\mathcal{S}$, if the policy $\mathbf{d}^\diamond$ is used, one can rewrite~\eqref{eqn:cost-func} as
\begin{align*}
\mathpzc C(\theta_k, h_k, a_k) = \alpha \bar{u} + (1-q(\bar{u},h_k)) \mathrm{Tr}(\mathtt{Var}(\theta_k))
\end{align*}
and for any $k\geq 1$,
\begin{align*}
\mathtt{Var}(\theta_{k+1})=&\left\{
        \begin{array}{ll}
            A\mathtt{Var}(\theta_k)A^\top + W, & \text{if $\gamma_k = 0$}, \\
            W, & \text{otherwise},
        \end{array}
    \right.
\end{align*}
with $\mathbb{P}(\gamma_k=0) = 1-q(\bar{u},h_k)$ and $\mathtt{Var}(\theta_1) = A\Sigma_0A^\top + W $. Then for any initial state $(\theta,h)\in\mathcal{S}$, with Assumption~\ref{asmpt:assumption-stability}, there exists a finite upper bound $\kappa(\theta)$, which depends on the initial state $\theta$, such that for any $k\geq 1$, $\mathbb{E}_{\theta,h}^{\mathbf{d^\diamond}}[\mathrm{Tr}(\mathtt{Var}(\theta_k))] \leq \kappa(\theta)$. Then one obtains that
\begin{align*}
\inf_{\mathbf{d},\theta,h}\mathpzc{J}(\mathbf{d},\theta,h) \leq&  \inf_{\theta,h}\mathpzc{J}(\mathbf{d}^{\diamond},\theta,h)  \\
<& \inf_{\theta} \kappa(\theta) + \alpha \bar{u} \\
<& \infty.
\end{align*}

We now focus on the verification of~\eqref{eqn:ExistConditionB}. Define the stopping time
\begin{align*}
\mathbb{T}_\beta \triangleq \inf\{k\geq 1: \upsilon_\beta(\theta_k,h_k) \leq \underline{\upsilon}_\beta(\mathscr{N}_{0,W})\},
\end{align*}
where $\underline{\upsilon}_\beta(\mathscr{N}_{0,W}) = \min_{h\in\mathbbm{h}} \upsilon_\beta(\mathscr{N}_{0,W},h)$. Then by Lemma~4.1 in~\cite{schal1993average}, one has for any $\beta<1$ and $(\theta,h) \in \mathcal{S}$,
\begin{align*}
\upsilon_\beta(\theta,h)-m_\beta \leq& \underline{\upsilon}_\beta(\mathscr{N}_{0,W}) - m_\beta \\
&\:+\, \inf_{\mathbf{d}\in\mathcal{D}}\mathbb{E}_{\theta,h}^{\mathbf{d}}
\left[\sum_{k=1}^{\mathbb{T}_\beta -1} \mathpzc C(\theta_k, h_k, a_k)\right]. \addtag \label{eqn:ExistConB_1}
\end{align*}

We now prove the finiteness of $\mathbb{E}_{\theta,h}^{\mathbf{d^\diamond}}[\mathbb{T}_\beta]$ with any initial state $(\theta,h)$. To this end, let $h^*= \argmin_{h\in\mathbbm{h}} \upsilon_\beta(\mathscr{N}_{0,W},h)$ and
\begin{align*}
\mathbb{T}_{\beta}^* \triangleq \inf\{k\geq 1: (\theta_k,h_k) = (\mathscr{N}_{0,W},h^*)\}.
\end{align*}
Note that the dependence of $\mathbb{T}_{\beta}^*$ on $\beta$ is due to $h^*$.
Then one can see that
for any realization of $\{\theta_k\}$ and $\{h_k\}$, \[\mathbb{T}_{\beta}^*\geq\mathbb{T}_{\beta}\] always holds. Note that $\{h_k\}$ evolves independently. Though $\{\theta_k\}$ depends on  the realization of $\{h_k\}$, under the policy $\mathbf{d}^{\diamond}$, $\mathbb{P}(\theta_k = \mathscr{N}_{0,W}) \geq q(\bar{u},\underline{h})$ for all $k>1$ with any initial state~$\theta$. Based on the above two observations, we construct a uniform (for any $0<\beta<1$) upper bound of $\mathbb{E}_{\theta,h}^{\mathbf{d^\diamond}}[\mathbb{T}_{\beta}^*]$ as follows.
Define
\[ \mathpzc K(h,h') = \min\{k>1: h_k=h', h_1 = h\}  \]
as the first time $h_k$ reaches $h'$ when starting at $h$.
Then given the initial state $h$, let $\{T_i\}_{i\geq 1}$ be a sequence of independent random variables such that $\mathbb{E}[T_1] = \mathbb{E}[\mathpzc K(h,h^*)]$ and $\mathbb{E}[T_i] = \mathbb{E}[\mathpzc K(h^*,h^*)], i>1$. Let $\chi$ be a geometrically distributed random variable with success probability $q(\bar{u},\underline{h})$. Then one obtains that
\begin{align*}
&\mathbb{E}_{\theta,h}^{\mathbf{d^\diamond}}[\mathbb{T}_{\beta}^*]
\\&\leq  \mathbb{E}\left[\sum_{i=1}^{\chi}T_i\right] \\
&\leq \frac{1}{q(\bar{u},\underline{h})} \max\{\mathbb{E}[\mathpzc K(h,h^*)],\mathbb{E}[\mathpzc K(h^*,h^*)]\} \\
&\leq \frac{1}{q(\bar{u},\underline{h})} \max_{h,h'\in\mathbbm{h}}\{\max\{\mathbb{E}[\mathpzc K(h,h')],\mathbb{E}[\mathpzc K(h',h')]\}\}
\addtag \label{eqn:UniformUpper}   \\
&< \infty, \addtag \label{eqn:UniformBounded}
\end{align*}
where the second inequality follows from the Wald's identity and the last inequality follows from the assumption that $\mathbbm{h}$ is a finite set and Assumption~\ref{asmpt:CommunicationModel}-(\rmnum{5}). Note that since~\eqref{eqn:UniformUpper} is independent of $\beta$, $\mathbb{E}_{\theta,h}^{\mathbf{d^\diamond}}[\mathbb{T}_{\beta}^*] $ is uniformly bounded.

With a little abuse of notation, define
\begin{align} \label{eqn:firstTimeStand}
 \mathpzc K(\theta,\theta') = \min\{k>1: \theta_k=\theta', \theta_1 = \theta\}.
\end{align}
Then for any initial state $(\theta,h)$, the finiteness of $\mathpzc{J}(\mathbf{d},\theta,h)$ implies $\mathbb{P}\left(\mathpzc K(\theta,\mathscr{N}_{0,W})<\infty\right) = 1$. Then by the definition of $\underline{\upsilon}_\beta(\mathscr{N}_{0,W})$, one obtains that
\begin{align} \label{eqn:ExistconBlimit0}
\limsup_{\beta \uparrow 1} (\underline{\upsilon}_\beta(\mathscr{N}_{0,W}) - m_\beta) = 0.
\end{align}
One thus obtains that for any $(\theta,h)\in\mathcal{S}$,
\begin{align*}
&\limsup_{\beta \uparrow 1}[\upsilon_\beta(\theta,h)-m_\beta] \\
& \leq   \limsup_{\beta \uparrow 1}  \inf_{\mathbf{d}\in\mathcal{D}}\mathbb{E}_{\theta,h}^{\mathbf{d}}\left[\sum_{k=1}^{\mathbb{T}_\beta -1} \mathpzc C(\theta_k, h_k, a_k)\right] \\
& \leq   \limsup_{\beta \uparrow 1}  \mathbb{E}_{\theta,h}^{\mathbf{d}^{\diamond}}\left[\sum_{k=1}^{\mathbb{T}_\beta -1} \mathpzc C(\theta_k, h_k, a_k)\right] \\
& \leq \limsup_{\beta \uparrow 1}   \mathbb{E}_{\theta,h}^{\mathbf{d}^{\diamond}}\left[\sum_{k=1}^{\mathbb{T}_{\beta}^* -1} \mathpzc C(\theta_k, h_k, a_k)\right] \\
& \leq \limsup_{\beta \uparrow 1} \mathbb{E}_{\theta,h}^{\mathbf{d^\diamond}}[\mathbb{T}_{\beta}^*-1] (\kappa(\theta) + \alpha \bar{u}) \\
& < \infty,
\end{align*}
where the first inequality follows from~\eqref{eqn:ExistConB_1}~and~\eqref{eqn:ExistconBlimit0}, the last second inequality follows from the Wald's identity and the last inequality follows from~\eqref{eqn:UniformBounded}.  The condition (relative discounted value function) thus is verified.

The proof of Theorem~\ref{theorem:ExistenceStationary} now is complete.


\section{Structural Description: Majorization Interpretation } \label{sec:structureResults}
In this section, we borrow the technical reasoning from~\cite{hajek2008paging, nayyar2013optimal} to show that the optimal transmission power allocation strategy has a symmetric and monotonic structure and the optimal estimator has a simple form for cases where the system is scalar.

Before presenting the main theorem, we introduce a notation as follows. For a policy $\mathbf{d}(d)\in\mathcal{D}_{\rm ds}$ with $d(\theta,h) = a(e)$, with a little abuse of notations, we write $a(e)$ as $a_{\theta,h}(e)$ to emphasize its dependence on the state $(\theta,h)$. We also use $a_{\theta,h}(e)$ to represent the deterministic and stationary policy $\mathbf{d}(d)$ with $d(\theta,h) = a(e)$.

According to Theorem~\ref{theorem:ExistenceStationary}, to
solve Problem~\ref{problem:2}, we can restrict
the optimal policy to be deterministic and stationary without
any performance loss. The following theorem suggests that
we the optimal policy can be further restricted to be
a specific class of functions.
{\begin{theorem} \label{theorem:StructuralResults}
Let the system~\eqref{eqn:process-dynamics} be scalar. There exists an optimal deterministic and stationary policy $a^*_{\theta,h}(e)$ such that
$a^*_{\theta,h}(e)$ is a symmetric and monotonic function of $e$, i.e., for any given $(\theta,h)\in\mathcal{S}$,
\begin{inparaenum}
\item[$(i)$.]$a^*_{\theta,h}(e) = a^*_{\theta,h}(-e)$
for all $e\in \mathbb R$;
\item[$(ii)$.]$a^*_{\theta,h}(e_1) \geq a^*_{\theta,h}(e_2)$ when $|e_1|\geq |e_2|$ with equality for $|e_1| = |e_2|$.
\end{inparaenum}
\end{theorem}

 The proof is given in Section~\ref{section:proofStructure}.} Note that Theorem~\ref{theorem:StructuralResults} does not require a symmetric initial distribution $\mu_{x_0}$. Intuitively, this is because 1) whatever the initial distribution is, the belief state will reach the very special state $\mathscr{N}_{0,W}$ sooner or later, 2) we focus on the long term average cost and the cost incurred by finite transient states can be omitted.

\begin{remark} \label{remark-FeedbackStructure}
When there exists only a finite number of power levels, only the thresholds of the innovation error used to switch the power levels are to be determined for computation of the optimal transmission power control strategy. This significantly simplifies both the offline computation complexity and the online implementation. The feedback of an accurate action $l_k$ in Figure~\ref{implementation-model} is possible in this scenario, since one  just need to feed back a finite number of thresholds, at which the power level switches.
\end{remark}
{
In the following theorem, without a proof, we give the optimal estimator~\eqref{eqn:mmse-x-k} when the transmission power controller adopts the optimal policy defined in Theorem~\ref{theorem:StructuralResults}. The simple structure of the optimal remote estimator $g_k^{\ast}$ in~\eqref{eqn:OptEstimatorStructure} is due to the symmetric structure the function $a^*_{\theta,h}(e)$ possesses.
Recall that $\tau(k)$ is defined in~\eqref{eqn:tauk}.
\begin{theorem}
Consider the optimal transmission power controller $f_k^*$,
\begin{equation*}
u_k =f_k^*(x_k,\mathcal O_k^-)\triangleq a^*_{\theta_k,h_k}(e_k)
\end{equation*}
where $a^*_{\theta,h}(e)$ is a symmetric and monotonic function of
$e_k$.
Then
the optimal
remote state estimator $g_k^*$ is given by
\begin{align} \label{eqn:OptEstimatorStructure}
\hat{x}_k=g_k^\ast(\mathcal{O}^+_k)= \left\{\begin{array}{ll}{x}_k,
 & \mathrm{if}~\gamma_k = 1, \\
  A^{\tau(k)}\hat x_{k-\tau(k)}, & \mathrm{if}~\gamma_k = 0,\end{array}\right.
\end{align}
\end{theorem}
In the following, we focus on the proof of Theorem~\ref{theorem:StructuralResults}.}
\subsection{Technical Preliminaries}
We first give some supporting definitions and lemmas as follows.
\begin{definition}[Symmetry]
A function $f:\mathbb{R}^n\to \mathbb{R}$ is said to be
symmetric about a point $o\in\mathbb{R}^n$, if, for
any two points $x,y\in\mathbb{R}^n$, $\|y-o\|=\|x-o\|$ implies
$f(x)=f(y)$.

\end{definition}

\begin{definition}[Unimodality]
A function $f:\mathbb{R}^n\to \mathbb{R}$ is said to be
unimodal, if there exists $o\in\mathbb{R}^n$ such that
$f(o)\geq f(o+\alpha_{0}v)\geq
f(o+\alpha_{1}v)$ holds for any
$v\in\mathbb{R}^n$ and any $\alpha_1\geq \alpha_0\geq0$.
\end{definition}

\begin{definition}
For any given Borel measurable set $\mathcal{B}\subset\mathbb{R}^n$, where $\mathfrak{L}(\mathcal{B})<\infty$, we denote
the symmetric rearrangement of $\mathcal{B}$ by $\mathcal{B}^\sigma$, i.e., $\mathcal{B}^\sigma$ is a ball centered at $0$ with the Lebesgue measure $\mathfrak{L}(\mathcal{B})$.
For a given integrable, nonnegative function $f:\mathbb{R}^n\to
\mathbb{R}$, we denote the symmetric nonincreasing rearrangement
 of $f$ by $f^{\sigma}$,
where $f^\sigma$ is defined as
$$f^\sigma(x)\triangleq \int_{0}^\infty \mathds{1}_{\{o\in\mathbb{R}^n:f(o)>t\}^{\sigma}}(x){\rm d}t.$$
\end{definition}

\begin{definition}
{
For any given two integrable, nonnegative functions
$f,g:\mathbb{R}^n\to \mathbb{R}$, we say that
$f$ majorizes $g$, which is denoted as $g\prec f$, if the following
conditions hold:
\begin{align}
\int_{\|x\|\leq t}f^{\sigma}(x){\rm d}x
\geq \int_{\|x\|\leq t}g^{\sigma}(x){\rm d}x~~~\forall t\geq 0  \label{eqn:MajorizationCond1}
\end{align}
and
$$\int_{\mathbb{R}^n}f(x){\rm d}x
= \int_{\mathbb{R}^n}g(x){\rm d}x.$$
Equivalently,~\eqref{eqn:MajorizationCond1} can be  altered by the following condition: for any Borel set $\mathcal{B} \subset \mathbb{R}^n$, there always exists another Borel set $\mathcal{B}'$ with $\mathfrak{L}(\mathcal{B}') =  \mathfrak{L}(\mathcal{B})$ such that $\int_{\mathcal{B}}g(x){\rm d}x
\leq \int_{\mathcal{B}'}f(x){\rm d}x.$
}
\end{definition}

Recall that $L$, which is introduced in~\eqref{eqn:StrucActionSatur}, is the saturation threshold for actions.
\begin{definition}[Binary Relation $\mathcal{R}$ on Belief States]
For any two belief states $\theta, \theta_\ast\in\Theta$, we say that
$\theta\mathcal{R}\theta_\ast$ if the following conditions hold:
\begin{enumerate}
\item[(\rmnum{1}).] there holds $\theta\prec \theta_*$;

\item[(\rmnum{2}).] $\theta_\ast$ is symmetric and unimodal about the origin point $0$.

\item[(\rmnum{3}).] $\theta(e) = \theta_\ast(e)$ for any $e\in \mathbb{R}^n\backslash\mathcal{E}$, where $\mathcal{E} \triangleq \{e\in\mathbb{R}^n: \|e\| \leq L\}$ is defined below~\eqref{eqn:StrucActionSatur}.
\end{enumerate}
\end{definition}

In the following, we define a symmetric increasing rearrangement of an action $a\in\mathcal{A}$, which preserves the average power consumption and successful transmission probability.
\begin{definition}
For any given Borel measurable $\mathcal{B}\subset\mathbb{R}^n$, where $\mathfrak{L}(\mathcal{B})<\infty$, we define
\[\mathcal{B}^{\sigma}_{\theta,\hat{\theta}} \triangleq \{e\in\mathbb{R}^n: \|e\|\geq r \}    ,\]
where $\theta,\hat{\theta}\in\Theta$ and $r$ is determined such that $\int_{\mathcal{B}}\theta(e)\mathrm{d}e = \int_{\mathcal{B}^{\sigma}_{\theta,\hat{\theta}}}\hat{\theta}(e)\mathrm{d}e$. Given an action $a\in\mathcal{A}$, define
$$a^\sigma_{\theta,\hat{\theta}}(e)\triangleq \int_{0}^\infty \mathds{1}_{\{o\in\mathbb{R}^n:a(o)>t\}^{\sigma}_{\theta,\hat{\theta}}}(e){\rm d}t.$$
\end{definition}
It can be verified that if $\int_{\mathbb{R}^n\backslash\mathcal{E}} \theta(e)\mathrm{d}e = \int_{\mathbb{R}^n\backslash\mathcal{E}} \hat{\theta}(e)\mathrm{d}e $,
$a^\sigma_{\theta,\hat{\theta}}(e)\in\mathcal{A}$. One also obtains that
\begin{align} \label{eqn:PowerPreserve}
\int_{\mathbb{R}^n} a(e) \theta(e) \mathrm{d}e =& \int_{\mathbb{R}^n} a^\sigma_{\theta,\hat{\theta}}(e) \hat{\theta}(e) \mathrm{d}e
\end{align}
and for any $h$,
\begin{align*}
\int_{\mathbb{R}^n} q(a(e),h) \theta(e) \mathrm{d}e =& \int_{\mathbb{R}^n} q\left(a^\sigma_{\theta,\hat{\theta}}(e),h\right) \hat{\theta}(e) \mathrm{d}e.
\end{align*}
Then the following lemma follows straightforwardly.
\begin{lemma} \label{lemma:Rrelationupdate}
If $A$ is a scalar, then $\theta\mathcal{R}\theta_\ast$ implies  $\phi(\theta, h,a, 0)\mathcal{R}\phi(\theta_\ast, h, a^\sigma_{\theta,\theta_\ast}, 0)$, where $\phi(\cdot,\cdot,\cdot,\cdot)$ is the belief update equation defined in~\eqref{eqn:BliefStateUpdate}.
\end{lemma}

Note that if $\theta\mathcal{R}\hat{\theta}$, then $q(a(e),h) \theta(e) \mathcal{R} q\left(a^\sigma_{\theta,\hat{\theta}}(e),h\right) \hat{\theta}(e)$. Then based on~\eqref{eqn:PowerPreserve}, following the same reasoning as in Lemma 15 in~\cite{nayyar2013optimal}, one obtains the following lemma.
\begin{lemma} \label{lemma:RrelationCost}
If $\theta\mathcal{R}\hat{\theta}$, then the following inequality about the one stage cost holds: $\mathpzc C(\theta, h, a) \geq \mathpzc C(\hat{\theta}, h, a^\sigma_{\theta,\hat{\theta}})$.
\end{lemma}

\subsection{Proof of Theorem~\ref{theorem:StructuralResults}}\label{section:proofStructure}
We then proceed to prove Theorem~\ref{theorem:StructuralResults} in a constructive way. To be specific, we show that for any initial state $(\theta,h)$, and any deterministic and stationary policy $\mathbf{d}(d)\in\mathcal{D}_{\rm ds}$\footnote{By Theorem~\ref{theorem:ExistenceStationary}, without any performance loss, we can just focus on the class of deterministic and stationary policies $\mathcal{D}_{\rm ds}$. } such that $\mathpzc{J}(\mathbf{d}(d),\theta,h) < \infty$\footnote{ Note that it is impossible for a policy with infinite cost to be optimal. }, there exists another policy $\hat{\mathbf{d}}(\hat{d})\in\mathcal{D}_{\rm ds}$ with a symmetric and monotonic structure defined in Theorem~\ref{theorem:StructuralResults} such that $\mathpzc{J}(\hat{\mathbf{d}}(\hat{d}),\theta,h)\leq\mathpzc{J}(\mathbf{d}(d),\theta,h)$. For any initial state $(\theta,h)$ and any policy with finite cost, $\mathbb{P}\left(\mathpzc K(\theta,\mathscr{N}_{0,W})<\infty\right) = 1$, where $\mathpzc K(\cdot,\cdot)$ is defined in~\eqref{eqn:firstTimeStand}. Then, without loss of generality, we can assume that the initial state $\theta = \mathscr{N}_{0,W}$. Let $d(\theta,h)=a_{\theta,h}(e)$, then under the policy $\mathbf{d}(d)$, the evolution of belief states is illustrated in Figure~\ref{fig:evolutionBlief}. Note that since the evolution of channel gains is independent of action $a$, we omit it. In Figure~\ref{fig:evolutionBlief}, the channel gain is assumed to be a constant $h$.  Notice the difference between the notations $\theta^i$ and $\theta_k$. The quantity $\theta^i$ denotes an element in $\Theta$, while $\theta_k$ is the belief state of the MDP at time instant $k$.
\begin{center}
\setlength{\unitlength}{1.8mm}
\begin{figure}[http]
\thicklines
\centering
\begin{picture}
(50,14)(-2,-9.5)
\thicklines

\put(5,0){\circle{6}}
{\large
\put(4.2,-1){$\theta^0$}}
\put(16,0){\circle{6}}
{\large
\put(15.2,-1){$\theta^1$}}
{\huge
\put(22,0){$\ldots$}
\put(48,0){$\ldots$}
}
\qbezier(2.5,2)(-3,0)(2.5,-2)
\put(2.8,-2.2){\vector(2,-1){0.1}}

\qbezier(5.6,2.8)(10.5,6)(14.8,3.3)
\put(15.2,2.6){\vector(1,-1){0.1}}

\qbezier(15.2,-2.6)(12,-5)(7.9,-2.3)
\put(7.4,-1.6){\vector(-1,1){0.1}}

\put(32,0){\circle{6}}
{\large
\put(31.2,-1){$\theta^i$}}

\qbezier(30.5,-2.6)(18,-9)(7.1,-3)
\put(6.6,-2.4){\vector(-1,1){0.1}}

\put(43,0){\circle{6}}
{\large
\put(41.2,-1.5){$\theta^{i+1}$}}

\qbezier(40.8,-2.6)(22,-12)(5.8,-3.5)
\put(5.4,-3.0){\vector(-1,1){0.1}}

\qbezier(32.6,2.8)(37.5,6)(41.7,3.3)
\put(42.2,2.6){\vector(1,-1){0.1}}

\put(-1.5,2){$p_0$}
\put(7,5.5){$1-p_0$}
\put(34,5.5){$1-p_i$}
\put(10,-2){$p_1$}
\put(21,-4.5){$p_i$}
\put(26,-6){$p_{i+1}$}

\end{picture}
\caption{Evolution of belief states. The special state $\theta^0 = \mathscr{N}_{0,W}$, $p_i = \varphi(\theta^i, h,a_{\theta^i,h}), \forall i\geq 0$
is the successful transmission probability, and $\theta^{i+1} = \phi(\theta^i, h,a_{\theta^i,h}, 0), \forall i\geq0$. When the belief state is $\theta^i$, it incurs cost $\mathpzc C(\theta^i,h,a_{\theta^i,h})$.}
\label{fig:evolutionBlief}
\end{figure}
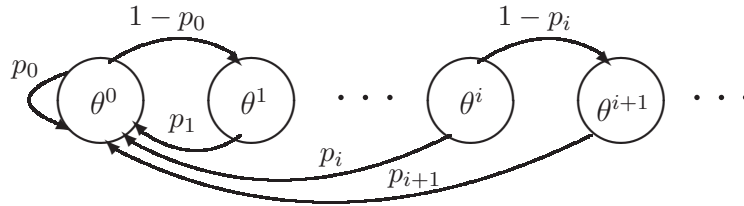
\end{center}
\vspace*{-7mm}
Let $\hat{d}(\theta,h)= \hat{a}_{\theta,h}(e)$, and $\hat{p}_i$ and $\hat{\theta}^i$ be the counterparts of $p_i$ and $\theta^i$, respectively. To facilitate presentation, let $a^i\triangleq a_{\theta^i,h}$ and $\hat{a}^i = \hat{a}_{\hat{\theta}^i,h}$. Then $\{\hat{a}^i\}_{i\in\mathbb N}$ are constructed as follows:
\[\hat{a}^i = \left(a^i\right)^\sigma_{\theta^i,\hat{\theta}^i}.\]
Then by Lemmas~\ref{lemma:Rrelationupdate}~and~\ref{lemma:RrelationCost}, one obtains that
\begin{align*}
\hat{p}_i = & p_i, \forall i\geq0, \\
\hat{\theta}^0 = \theta^0 = \mathscr{N}_{0,W},& \quad  \theta^i\mathcal{R}\hat{\theta}^i, \forall i\geq 1, \\
\mathpzc C(\theta^i,h,a^i) \geq& \mathpzc C(\hat{\theta}^i,h,\hat{a}^i),\forall i\geq0, h\in\mathbbm{h}.
\end{align*}
It then follows that $\mathpzc{J}(\hat{\mathbf{d}}(\hat{d}),\theta,h)\leq\mathpzc{J}(\mathbf{d}(d),\theta,h)$. Since $\{\hat{a}^i\}_{i\in \mathbb N}$ is symmetric and increasing, and $\hat{\theta}^i$ is symmetric, one concludes the results of the theorem.

\section{Conclusion and Future Work}  \label{sec:Conclusion}
In this paper, we studied the remote estimation problem where the sensor communicates with the remote estimator over a fading channel. The transmission power control strategy, which affects the behavior of communications, as well as the remote estimator were optimally co-designed to minimize an infinite horizon cost consisting of power consumption and estimation error. We showed that when determining the optimal transmission power, the full information history available at the sensor is equivalent to its belief state. Since no constraints on the information structure are imposed and the belief state is updated recursively, the results we obtained provide some insights into the qualitative characterization of the optimal power allocation strategy and facilitate the related performance analyses. In particular, we provided some structural results on the optimal power allocation strategy and the optimal estimator, which simplifies the practical implementation of the algorithm significantly.
One direction of future work is to explore the structural description of the optimal remote estimator and the optimal transmission power control rule
when the system matrix is a general one.

%
%
%
%
%

\section*{Appendix A }
We will introduce background knowledge for weak convergence of probability measures and generalized skorohad space.
%
%
%

\subsection{Weak Convergence of Probability Measures}
Let $\mathcal X$ be a general Polish space and $\mathscr{X}$ be the Borel $\sigma-$field~\cite{durrett2010probability}. Let $\mu$ and $\{\mu_{i, i\in\mathbb{N}}\}$ be probability measures on $(\mathcal X,\mathscr{X})$.
By the Portmanteau Theorem~\cite{billingsley2013convergence}, the following statements are equivalent:
\begin{enumerate}
  \item[(\rmnum{1}).] $\mu_i$ converges weakly to $\mu$.
  \item[(\rmnum{2}).] $\lim_{i\to\infty}\int b \mathrm{d} \mu_i  \rightarrow  \int b \mathrm{d} \mu$ for every bounded and continuous function $b(\cdot)$ on $X$.
  \item[(\rmnum{3}).] $\lim_{i\to\infty}\mu_i(\mathcal B) \to \mu(\mathcal B)$ for every $\mu-$continuity set $\mathcal B$.
\end{enumerate}
We write as $\mu_i \overset{w}{\to} \mu $ if $\mu_i$ converges weakly to $\mu$. The Prohorov metric~\cite{billingsley2013convergence} is a  metrization of this weak convergence topology.
Let $\mathcal{M}$ be the collection of all the probability measures defined on $(\mathcal X,\mathscr{X})$. If $\mathcal{M}$ is endowed with the weak convergence topology, then $\mathcal{M}$ is a Polish space.

\subsection{Generalized Skorohod Space~\cite{straf1972}.}
Let $(\mathcal X,\mathbbm{d}_{\mathcal X}(\cdot,\cdot))$ be a compact metric space and $\Lambda$ be a set of homeomorphisms from $\mathcal X$ onto itself. Let $\bbpi$ be a generic element of $\Lambda$, then on $\Lambda$, define the following three norms:
\begin{align*}
\|\bbpi\|_{\rm s} =& \sup_{x\in \mathcal X} \mathbbm{d}_{\mathcal X}(\bbpi x, x)\\
\|\bbpi\|_{\rm t} =& \sup_{x,y\in \mathcal X: x\not=y}\Big|\log \frac{\mathbbm{d}_{\mathcal X}(\bbpi x, \bbpi y)}{\mathbbm{d}_{\mathcal X}( x, y)}\Big| \\
\|\bbpi\|_{\rm m} =& \|\bbpi\|_{\rm s} +  \|\bbpi\|_{\rm t}.
\end{align*}
Note that $\|\bbpi\|_{\rm t} = \|\bbpi^{-1}\|_{\rm t}$. Let $\Lambda_{\rm t} \subseteq \Lambda$ be the group of homeomorphisms with finite $\|\cdot\|_{\rm t}$, i.e.,
\[\Lambda_{\rm t} = \{\bbpi\in\Lambda:  \|\bbpi\|_{\rm t} < \infty  \}.\]
Note that since $\mathcal X$ is compact, each element in $\Lambda_{\rm t}$ also has finite $\|\cdot\|_{\rm m}$.
Let $\mathcal{B}_r(\mathcal X)$ be the set of bounded real-valued functions defined on $\mathcal X$, then the Skorohod distance $\mathbbm{d}(\cdot,\cdot)$ for $f,g\in\mathcal{B}_r(\mathcal X)$ is defined by
\begin{align*}
\mathbbm{d}(f,g) = \inf_{\epsilon}\{&\epsilon > 0: \exists\, \bbpi\in\Lambda_{\rm t} \:\text{such that} \\
& \qquad  \|\bbpi\|_{\rm m} < \epsilon \:\text{and}\: \sup_{x\in \mathcal X}|f(x)-g(\bbpi x)| < \epsilon \}. \addtag \label{eqn:SkorohodDistance}
\end{align*}
Let $\mathcal{W}$ be the set of all finite partitions of $\mathcal X$ that are invariant under $\Lambda$. Let $I_{\Delta}$ be the collection of functions that are constant on each cell of  a partition $\Delta\in\mathcal {W}$. Then the generalized Skorohod space on $\mathcal X$ are defined by
\begin{align*}
\mathcal{D}(\mathcal X) = \{ f\in \mathcal{B}_{r}(\mathcal X):&\:\exists\:\Delta \in \mathcal{W},\:
g\in I_{\Delta}\:\text{such that} \\
& \quad \mathbbm{d}(f,g) = 0     \}. \addtag \label{eqn:SkorohodSpace}
\end{align*}
By convention, two functions $f$ and $g$ with $\mathbbm{d}(a,b)=0$ are not distinguished.
Then by Lemma 3.4, Theorems 3.7 and~3.8 in~\cite{straf1972}, the space $\mathcal{D}(\mathcal X)$ of the resulting equivalence classes with metric $\mathbbm{d}(\cdot,\cdot)$ defined in~\eqref{eqn:SkorohodDistance} is a complete metric space.
For $f\in\mathcal{B}_{r}(X)$ and $\Delta =\{\delta_j\}\in \mathcal{W}$, define
\begin{align}
 \mathbbm{w}(f,\Delta) = \max_{\delta_j}\sup_{x,y}\{|f(x)-f(y)|: x,y\in\delta_j\}.  \label{eqn:SkorohodFuncw}
 \end{align}
For $f\in\mathcal{B}_{r}(X)$, $f\in\mathcal{D}(\mathcal X)$ if and only if $\lim_{\Delta}\mathbbm{w}(f,\Delta) \to 0,$ with the limits taken along the direction of refinements.

In the end, we should remark that in our case for the action space $\mathcal{A}$, the collection of homeomorphisms $\Lambda_{\rm t}$ used to define the Skorohod distance in \eqref{eqn:SkorohodDistance}  is defined on $\mathcal{E}$ instead of $\mathbb{R}^n$.

\section*{Appendix B}
In the following lemma, we give a condition on the probability measures, under which the weak convergence implies set-wise convergence.
\begin{lemma} \label{lemma:WeakConv2Setwise}
Let $\mu$ and $\{\mu_{i,i\in\mathbb{N}}\}$ be probability measures defined on $(\mathbb{R}^n, \mathscr {B}(\mathbb{R}^n))$, where
$\mathscr {B}(\mathbb{R}^n)$ denotes the Borel $\sigma-$algebra of $\mathbb{R}^n$. Suppose they are absolutely continuous with respect to the Lebesgue measure. Then the following holds:
\begin{align}
\mu_i \overset{w}{\to} \mu  \Rightarrow  \mu_i \overset{sw}{\to} \mu, \label{eqn:weak2setwise}
\end{align}
where $\mu_i \overset{sw}{\to} \mu$ represents set-wise convergence, i.e., for any $\mathcal A\in \mathscr {B}(\mathbb{R}^n)$, $\mu_i(\mathcal A) \to \mu(\mathcal A).$
\end{lemma}
\begin{proof}
The Borel $\sigma-$algebra $\mathscr {B}(\mathbb{R}^n)$ can be generated by $n-$demensional rectangles, i.e.,
\begin{align} \label{eqn:sigmaAlgebra}
\mathscr {B}(\mathbb{R}^n) = \sigma(\{(x_1,y_1]\times\cdots\times(x_n,y_n]: x_j,y_j\in\mathbb{R}\}).
\end{align}
Since $\mu$ is absolutely continuous with respect to Lebesgue measure, all the rectangles are $\mu-$continuity sets. By the Portmanteau Theorem~\cite{billingsley2013convergence}, for any $x_j,y_j\in\mathbb{R}$,
\[\mu_i((x_1,y_1]\times\cdots\times(x_n,y_n]) \to \mu((x_1,y_1]\times\cdots\times(x_n,y_n]).\]
Then statement~\eqref{eqn:weak2setwise} follows from~\eqref{eqn:sigmaAlgebra}, which completes the proof.
\end{proof}

\bibliographystyle{IEEETran}
\bibliography{reference1,dquevedo,sj_reference,reference_xq,sj_reference_1}

\begin{thebibliography}{10}
\providecommand{\url}[1]{#1}
\csname url@samestyle\endcsname
\providecommand{\newblock}{\relax}
\providecommand{\bibinfo}[2]{#2}
\providecommand{\BIBentrySTDinterwordspacing}{\spaceskip=0pt\relax}
\providecommand{\BIBentryALTinterwordstretchfactor}{4}
\providecommand{\BIBentryALTinterwordspacing}{\spaceskip=\fontdimen2\font plus
\BIBentryALTinterwordstretchfactor\fontdimen3\font minus
  \fontdimen4\font\relax}
\providecommand{\BIBforeignlanguage}[2]{{%
\expandafter\ifx\csname l@#1\endcsname\relax
\typeout{** WARNING: IEEEtran.bst: No hyphenation pattern has been}%
\typeout{** loaded for the language `#1'. Using the pattern for}%
\typeout{** the default language instead.}%
\else
\language=\csname l@#1\endcsname
\fi
#2}}
\providecommand{\BIBdecl}{\relax}
\BIBdecl

\bibitem{wong2}
W.~S. Wong and R.~W. Brockett, ``Systems with finite communication
  bandwidth-part {II}: Stabilization with limited information feedback,''
  \emph{IEEE Transactions on Automatic Control}, vol.~44, no.~5, pp.
  1049--1053, 1999.

\bibitem{nair2004stabilizability}
G.~N. Nair and R.~J. Evans, ``Stabilizability of stochastic linear systems with
  finite feedback data rates,'' \emph{SIAM Journal on Control and
  Optimization}, vol.~43, no.~2, pp. 413--436, 2004.

\bibitem{tatikonda2}
S.~Tatikonda and S.~Mitter, ``Control under communication constraints,''
  \emph{IEEE Transactions on Automatic Control}, vol.~49, pp. 1056--1068, July
  2004.

\bibitem{ishii}
H.~Ishii and B.~A. Francis, ``Quadratic stabilization of sampled-data systems
  with quantization,'' \emph{Automatica}, vol.~39, pp. 1793--1800, 2003.

\bibitem{fu-xie-tac05}
M.~Fu and L.~Xie, ``The sector bound approach to quantized feedback control,''
  \emph{IEEE Transactions on Automatic Control}, vol.~50, no.~11, pp.
  1698--1711, 2005.

\bibitem{joao07}
J.~Hespanha, P.~Naghshtabrizi, and Y.~Xu, ``A survey of recent results in
  networked control systems,'' \emph{Proceedings of the IEEE}, vol.~95, no.~1,
  pp. 138--162, 2007.

\bibitem{schenato2008optimal}
L.~Schenato, ``Optimal estimation in networked control systems subject to
  random delay and packet drop,'' \emph{IEEE Transactions on Automatic
  Control}, vol.~53, no.~5, pp. 1311--1317, 2008.

\bibitem{shi2009kalman}
L.~Shi, L.~Xie, and R.~M. Murray, ``Kalman filtering over a packet-delaying
  network: A probabilistic approach,'' \emph{Automatica}, vol.~45, no.~9, pp.
  2134--2140, 2009.

\bibitem{you2011mean}
K.~You, M.~Fu, and L.~Xie, ``Mean square stability for kalman filtering with
  markovian packet losses,'' \emph{Automatica}, vol.~47, no.~12, pp.
  2647--2657, 2011.

\bibitem{sinopoli2004kalman}
B.~Sinopoli, L.~Schenato, M.~Franceschetti, K.~Poolla, M.~I. Jordan, and S.~S.
  Sastry, ``Kalman filtering with intermittent observations,'' \emph{IEEE
  Transactions on Automatic Control}, vol.~49, no.~9, pp. 1453--1464, 2004.

\bibitem{huang-dey-stability-kf}
M.~Huang and S.~Dey, ``Stability of {K}alman filtering with markovian packet
  losses,'' \emph{Automatica}, vol.~43, pp. 598--607, 2007.

\bibitem{schenato2007foundations}
L.~Schenato, B.~Sinopoli, M.~Franceschetti, K.~Poolla, and S.~S. Sastry,
  ``Foundations of control and estimation over lossy networks,''
  \emph{Proceedings of the IEEE}, vol.~95, no.~1, pp. 163--187, 2007.

\bibitem{gupta2009optimal}
V.~Gupta, N.~C. Martins, and J.~S. Baras, ``Optimal output feedback control
  using two remote sensors over erasure channels,'' \emph{IEEE Transactions on
  Automatic Control}, vol.~54, no.~7, pp. 1463--1476, 2009.

\bibitem{mainwaring2002wireless}
A.~Mainwaring, D.~Culler, J.~Polastre, R.~Szewczyk, and J.~Anderson, ``Wireless
  sensor networks for habitat monitoring,'' in \emph{Proc. International
  Workshop on Wireless Sensor Networks and Applications}.\hskip 1em plus 0.5em
  minus 0.4em\relax ACM, 2002, pp. 88--97.

\bibitem{yang2011deterministic}
C.~Yang and L.~Shi, ``Deterministic sensor data scheduling under limited
  communication resource,'' \emph{IEEE Transactions on Signal Processing},
  vol.~59, no.~10, pp. 5050--5056, 2011.

\bibitem{zhao2014optimal}
L.~Zhao, W.~Zhang, J.~Hu, A.~Abate, and C.~J. Tomlin, ``On the optimal
  solutions of the infinite-horizon linear sensor scheduling problem,''
  \emph{IEEE Transactions on Automatic Control}, vol.~59, no.~10, pp.
  2825--2830, 2014.

\bibitem{shi2011optimal}
L.~Shi, P.~Cheng, and J.~Chen, ``Optimal periodic sensor scheduling with
  limited resources,'' \emph{IEEE Transactions on Automatic Control}, vol.~56,
  no.~9, pp. 2190--2195, 2011.

\bibitem{huber2012optimal}
M.~F. Huber, ``Optimal pruning for multi-step sensor scheduling,'' \emph{IEEE
  Transactions on Automatic Control}, vol.~57, no.~5, pp. 1338--1343, 2012.

\bibitem{shi2013optimal}
D.~Shi and T.~Chen, ``Optimal periodic scheduling of sensor networks: A branch
  and bound approach,'' \emph{Systems \& Control Letters}, vol.~62, no.~9, pp.
  732--738, 2013.

\bibitem{liu2014optimal}
S.~Liu, M.~Fardad, P.~K. Varshney, and E.~Masazade, ``Optimal periodic sensor
  scheduling in networks of dynamical systems,'' \emph{IEEE Transactions on
  Signal Processing}, vol.~62, no.~12, pp. 3055--3068, 2014.

\bibitem{aastrom2002comparison}
K.~J. {\AA}str{\"o}m and B.~Bernhardsson, ``Comparison of riemann and lebesque
  sampling for first order stochastic systems,'' in \emph{Proceedings of the
  41st IEEE Conference on Decision and Control}, vol.~2.\hskip 1em plus 0.5em
  minus 0.4em\relax IEEE, 2002, pp. 2011--2016.

\bibitem{xu05cdc}
Y.~Xu and J.~Hespanha, ``Estimation under uncontrolled and controlled
  communications in networked control systems,'' in \emph{Proc. IEEE Conference
  on Decision and Control and European Control Conference}, Dec 2005, pp.
  842--847.

\bibitem{Imer-CDC-05}
O.~Imer and T.~Ba\c{s}ar, ``Optimal estimation with limited measurements,'' in
  \emph{Proceedings of the 44th IEEE Conference on Decision and Control and
  European Control Conference}, December 2005, pp. 1029--1034.

\bibitem{randy-cogill-acc07}
R.~Cogill, S.~Lall, and J.~P. Hespanha, ``A constant factor approximation
  algorithm for event-based sampling,'' in \emph{Proceedings of the American
  Control Conference}, 2007, pp. 305--311.

\bibitem{sijs2009event}
J.~Sijs and M.~Lazar, ``On event based state estimation,'' in \emph{Hybrid
  Systems: computation and control}.\hskip 1em plus 0.5em minus 0.4em\relax
  Springer, 2009, pp. 336--350.

\bibitem{Lispa11TAC}
G.~Lipsa and N.~Martins, ``Remote state estimation with communication costs for
  first-order {LTI} systems,'' \emph{IEEE Transactions on Automatic Control},
  vol.~56, no.~9, pp. 2013--2025, Sept 2011.

\bibitem{wu2013event}
J.~Wu, Q.-S. Jia, K.~H. Johansson, and L.~Shi, ``Event-based sensor data
  scheduling: Trade-off between communication rate and estimation quality,''
  \emph{IEEE Transactions on Automatic Control}, vol.~58, no.~4, pp.
  1041--1046, 2013.

\bibitem{nayyar2013optimal}
A.~Nayyar, T.~Ba\c{s}ar, D.~Teneketzis, and V.~V. Veeravalli, ``Optimal
  strategies for communication and remote estimation with an energy harvesting
  sensor,'' \emph{IEEE Transactions on Automatic Control}, vol.~58, no.~9, pp.
  2246--2260, 2013.

\bibitem{ramesh2013design}
C.~Ramesh, H.~Sandberg, and K.~H. Johansson, ``Design of state-based schedulers
  for a network of control loops,'' \emph{IEEE Transactions on Automatic
  Control}, vol.~58, no.~8, pp. 1962--1975, 2013.

\bibitem{junfeng14tac}
J.~Wu, K.~Johansson, and L.~Shi, ``A stochastic online sensor scheduler for
  remote state estimation with time-out condition,'' \emph{IEEE Transactions on
  Automatic Control}, vol.~59, no.~11, pp. 3110--3116, 2014.

\bibitem{molin2014optimal}
A.~Molin, ``Optimal event-triggered control with communication constraints,''
  Ph.D. dissertation, M{\"u}nchen, Technische Universit{\"a}t M{\"u}nchen,
  Diss., 2014, 2014.

\bibitem{vijay_sensor_schedule}
V.~Gupta, T.~Chung, B.~Hassibi, and R.~M. Murray, ``On a stochastic sensor
  selection algorithm with applications in sensor scheduling and dynamic sensor
  coverage,'' \emph{Automatica}, vol.~42, no.~2, pp. 251--260, 2006.

\bibitem{han2014stochastic}
D.~Han, Y.~Mo, J.~Wu, S.~Weerakkody, B.~Sinopoli, and L.~Shi, ``Stochastic
  event-triggered sensor schedule for remote state estimation,'' \emph{IEEE
  Transactions on Automatic Control,to appear}, 2015.

\bibitem{journals/corr/WeerakkodyMSHS15}
S.~Weerakkody, Y.~Mo, B.~Sinopoli, D.~Han, and L.~Shi, ``Multi-sensor
  scheduling for state estimation with event-based, stochastic triggers,''
  \emph{IEEE Transactions on Automatic Control}, accepted.

\bibitem{yuksel2013stochastic}
S.~Y{\"u}ksel and T.~Ba{\c{s}}ar, \emph{Stochastic networked control systems:
  Stabilization and optimization under information constraints}.\hskip 1em plus
  0.5em minus 0.4em\relax Springer Science \& Business Media, 2013.

\bibitem{goldsmith2005wireless}
A.~Goldsmith, \emph{Wireless {C}ommunications}.\hskip 1em plus 0.5em minus
  0.4em\relax Cambridge {U}niversity {P}ress, 2005.

\bibitem{elia2005remote}
N.~Elia, ``Remote stabilization over fading channels,'' \emph{Systems \&
  Control Letters}, vol.~54, no.~3, pp. 237--249, 2005.

\bibitem{dey2009kalman}
S.~Dey, A.~S. Leong, and J.~S. Evans, ``Kalman filtering with faded
  measurements,'' \emph{Automatica}, vol.~45, no.~10, pp. 2223--2233, 2009.

\bibitem{xiao2012feedback}
N.~Xiao, L.~Xie, and L.~Qiu, ``Feedback stabilization of discrete-time
  networked systems over fading channels,'' \emph{IEEE Transactions on
  Automatic Control}, vol.~57, no.~9, pp. 2176--2189, 2012.

\bibitem{Quevedo13TAC}
D.~E. Quevedo, A.~Ahl\'{e}n, and K.~H. Johansson, ``State estimation over
  sensor networks with correlated wireless fading channels,'' \emph{IEEE
  Transactions on Automatic Control}, vol.~58, no.~3, pp. 581--593, 2013.

\bibitem{wang2009distortion}
C.-H. Wang and S.~Dey, ``Distortion outage minimization in {R}ayleigh fading
  using limited feedback,'' in \emph{IEEE Global Telecommunications Conference
  (GLOBECOM)}.\hskip 1em plus 0.5em minus 0.4em\relax IEEE, 2009, pp. 1--8.

\bibitem{queahl10}
D.~E. Quevedo, A.~Ahl\'{e}n, and J.~{\O stergaard}, ``Energy efficient state
  estimation with wireless sensors through the use of predictive power control
  and coding,'' \emph{{IEEE} Transactions Signal Processing}, vol.~58, no.~9,
  pp. 4811--4823, 2010.

\bibitem{leong2011power}
A.~S. Leong, S.~Dey, G.~N. Nair, and P.~Sharma, ``Power allocation for outage
  minimization in state estimation over fading channels,'' \emph{IEEE
  Transactions on Signal Processing}, vol.~59, no.~7, pp. 3382--3397, 2011.

\bibitem{WuAutomatica13}
J.~Wu, Y.~Li, D.~E. Quevedo, V.~Lau, and L.~Shi, ``Data-driven power control
  for state estimation: a {B}ayesian inference approach,'' \emph{Automatica},
  vol.~54, pp. 332--339, 2015.

\bibitem{gatsis2014optimal}
K.~Gatsis, A.~Ribeiro, and G.~J. Pappas, ``Optimal power management in wireless
  control systems,'' \emph{IEEE Transactions on Automatic Control}, vol.~59,
  no.~6, pp. 1495--1510, 2014.

\bibitem{nourian2014optimal}
M.~Nourian, A.~S. Leong, and S.~Dey, ``Optimal energy allocation for {K}alman
  filtering over packet dropping links with imperfect acknowledgments and
  energy harvesting constraints,'' \emph{IEEE Transactions on Automatic
  Control}, vol.~59, no.~8, pp. 2128--2143, 2014.

\bibitem{Nourian15JSAC}
M.~Nourian, S.~Dey, and A.~Ahl\'{e}n, ``Distortion minimization in multi-sensor
  estimation with energy harvesting,'' \emph{IEEE Journal on Selected Areas in
  Communications}, vol.~33, no.~3, pp. 524--539, 2015.

\bibitem{fu2009Automatica}
M.~Fu and C.~E. de~Souza, ``State estimation for linear discrete-time systems
  using quantized measurements,'' \emph{Automatica}, vol.~45, no.~12, pp. 2937
  -- 2945, 2009.

\bibitem{leong2012power}
A.~S. Leong and S.~Dey, ``{Power allocation for error covariance minimization
  in Kalman filtering over packet dropping links},'' in \emph{Proceedings of
  the 51st IEEE Conference on Decision and Control}.\hskip 1em plus 0.5em minus
  0.4em\relax IEEE, 2012, pp. 3335--3340.

\bibitem{mostofi2009drop}
Y.~Mostofi and R.~M. Murray, ``To drop or not to drop: design principles for
  {K}alman filtering over wireless fading channels,'' \emph{IEEE Transactions
  on Automatic Control}, vol.~54, no.~2, pp. 376--381, 2009.

\bibitem{haykin1988digital}
S.~S. Haykin, \emph{Digital {C}ommunications}.\hskip 1em plus 0.5em minus
  0.4em\relax New York: Wiley, 1988.

\bibitem{anderson79}
B.~Anderson and J.~Moore, \emph{{Optimal Filtering}}.\hskip 1em plus 0.5em
  minus 0.4em\relax Prentice Hall, 1979.

\bibitem{nayyar2011sequential}
A.~Nayyar, ``Sequential decision making in decentralized systems,'' Ph.D.
  dissertation, University of California, Berkeley, 2011.

\bibitem{Cassandra1998POMDP}
A.~R. Cassandra, ``Exact and approximate for partially observable {M}arkov
  decision processes,'' Ph.D. dissertation, Brown University, 1998.

\bibitem{durrett2010probability}
R.~Durrett, \emph{Probability: {T}heory and {E}xamples}.\hskip 1em plus 0.5em
  minus 0.4em\relax Cambridge university press, 2010.

\bibitem{billingsley2013convergence}
P.~Billingsley, \emph{Convergence of {P}robability {M}easures}.\hskip 1em plus
  0.5em minus 0.4em\relax New York: John Wiley \& Sons, 1999.

\bibitem{schal1993average}
M.~Sch{\"a}l, ``Average optimality in dynamic programming with general state
  space,'' \emph{Mathematics of Operations Research}, vol.~18, no.~1, pp.
  163--172, 1993.

\bibitem{feinberg2012average}
E.~A. Feinberg, P.~O. Kasyanov, and N.~V. Zadoianchuk, ``Average cost {M}arkov
  decision processes with weakly continuous transition probabilities,''
  \emph{Mathematics of Operations Research}, vol.~37, no.~4, pp. 591--607,
  2012.

\bibitem{hernandez2012discrete}
O.~Hern{\'a}ndez-Lerma and J.~B. Lasserre, \emph{Discrete-time {M}arkov
  {C}ontrol {P}rocesses: {B}asic {O}ptimality {C}riteria}.\hskip 1em plus 0.5em
  minus 0.4em\relax New York: Springer Science \& Business Media, 2012,
  vol.~30.

\bibitem{MInyan2010Topology}
M.~Yan, \emph{Introduction to Topology: Theory and Applications}.\hskip 1em
  plus 0.5em minus 0.4em\relax Beijing: Higher Education Press, 2010.

\bibitem{lange1973borel}
K.~Lange, ``Borel sets of probability measures,'' \emph{Pacific Journal of
  Mathematics}, vol.~48, pp. 141--161, 1973.

\bibitem{feinberg2014fatou}
E.~A. Feinberg, P.~O. Kasyanov, and N.~V. Zadoianchuk, ``Fatou's lemma for
  weakly converging probabilities,'' \emph{Theory of Probability \& Its
  Applications}, vol.~58, no.~4, pp. 683--689, 2014.

\bibitem{rudin1964principles}
W.~Rudin, \emph{Principles of {M}athematical {A}nalysis}.\hskip 1em plus 0.5em
  minus 0.4em\relax New York: McGraw-Hill, 1964, vol.~3.

\bibitem{straf1972}
M.~L. Straf, ``Weak convergence of stochastic processes with several
  parameters,'' in \emph{Proceedings of the Sixth Berkeley Symposium on
  Mathematical Statistics and Probability, Volume 2: Probability Theory}.\hskip
  1em plus 0.5em minus 0.4em\relax Berkeley, Calif.: University of California
  Press, 1972, pp. 187--221.

\bibitem{hernandez2000fatou}
O.~Hern{\'a}ndez-Lerma and J.~B. Lasserre, ``Fatou's lemma and {L}ebesgue's
  convergence theorem for measures,'' \emph{International Journal of Stochastic
  Analysis}, vol.~13, no.~2, pp. 137--146, 2000.

\bibitem{hajek2008paging}
B.~Hajek, K.~Mitzel, and S.~Yang, ``Paging and registration in cellular
  networks: Jointly optimal policies and an iterative algorithm,'' \emph{IEEE
  Transactions on Information Theory}, vol.~54, no.~2, pp. 608--622, 2008.

\end{thebibliography}

\end{document}